\documentclass[reqno,10pt]{amsart}

\usepackage{amsmath, amsthm, amsfonts}
\usepackage{amssymb}
\usepackage{pdfsync}
\usepackage{enumitem}
\usepackage{hyperref}
\usepackage{mathrsfs} 
\usepackage[british]{babel} 
\usepackage{mathtools}
\mathtoolsset{showonlyrefs}
\usepackage{mdframed}
\usepackage{graphicx}
\usepackage{wrapfig}
\usepackage{tikz}
\usepackage{ctable}
\usepackage{caption}
\usepackage{subcaption}
\usepackage[margin=1in]{geometry}
\usepackage[sort]{cite}

\numberwithin{equation}{section}
\newtheorem{theorem}{Theorem}[section]
\newtheorem{prop}[theorem]{Proposition}

\newtheorem{lem}[theorem]{Lemma}

\newtheorem{example}[theorem]{Example}
\theoremstyle{remark}
\newtheorem{remark}[theorem]{Remark}

\newenvironment{subproof}[1][\proofname]
{%
  \begin{proof}[#1]%
}
{%
  \end{proof}%
}

\allowdisplaybreaks

\newcommand{\cA}{\mathcal A}

\newcommand{\cC}{\mathcal C}

\newcommand{\cE}{\mathcal E}
\newcommand{\cF}{\mathcal F}
\newcommand{\cG}{\mathcal G}
\newcommand{\cH}{\mathcal H}

\newcommand{\cK}{\mathcal K}
\newcommand{\cL}{\mathcal L}
\newcommand{\cM}{\mathcal M}
\newcommand{\cN}{\mathcal N}

\newcommand{\cP}{\mathcal P}

\newcommand{\cS}{\mathcal S}

\newcommand{\cX}{\mathcal X}

\newcommand{\sL}{\mathsf{L}}

\def\R{\mathbb{R}}

\def\bE{\mathbb{E}}
\def\N{\mathbb{N}}

\def\min{{\rm min}}
\def\max{{\rm max}}

\def\Inp{I_{\rm np}}
\def\Ip{I_{\rm p}}

\def\Oi{\Omega_{{\rm i}}} 
\def\Oe{\Omega_{{\rm e}}}  
\def\Ot{\Omega_{{\rm t}}}

\def\um{u_{{\rm min}}}
\def\umn{u_{{\rm min};n}}

\def\sv{\mathsf{v}}
\def\a{\mathsf{a}} 
\def\Ipd{I_{{\rm p},md}}

\def\q{\mathsf{q}}

\def\srE{\mathscr{E}}

\def\Oig{\Omega_{{\rm i},\Gamma}} 
\def\Oeg{\Omega_{{\rm e},\Gamma}}  
\def\Otg{\Omega_{{\rm t},\Gamma}}
\def\Oign{\Omega_{{\rm i},\Gamma_n}} 
\def\Oegn{\Omega_{{\rm e},\Gamma_n}}

\begin{document}

\title[ensemble average energy and solute-solvent interfacial fluctuations]{Modeling and analysis of ensemble average solvation energy and solute-solvent interfacial fluctuations}

\author[]{
Yuanzhen Shao$^{* \dagger}$,
Zhan Chen$^{**}$, 
and
Shan Zhao$^{***}$ 
}

\thanks{$^{\dagger}$Corresponding author.}

\thanks{$^{*}$Department of Mathematics, The University of Alabama, Tuscaloosa, AL, USA.
\texttt{yshao8@ua.edu}}

\thanks{$^{**}$Department of Mathematical Sciences, Georgia Southern University, GA, USA.
\texttt{zchen@georgiasouthern.edu}}

\thanks{$^{***}$Department of Mathematics, The University of Alabama, Tuscaloosa, AL, USA.
\texttt{szhao@ua.edu}}

\normalsize

\begin{abstract}
Variational implicit solvation models (VISM) have gained extensive popularity in the molecular-level solvation analysis of biological systems due to their cost-effectiveness and satisfactory accuracy.
Central in the construction of VISM is an interface separating the solute and the solvent. 
However, traditional sharp-interface VISMs fall short in adequately representing the inherent randomness of the solute-solvent interface, a consequence of thermodynamic fluctuations within the solute-solvent system.
Given that experimentally observable quantities are ensemble-averaged, the computation of the ensemble average solvation energy (EASE)-the averaged solvation energy across all thermodynamic microscopic states-emerges as a key metric for reflecting thermodynamic fluctuations during solvation processes.
This study introduces a novel approach to calculating the EASE.
We devise two diffuse-interface VISMs: one within the classic Poisson-Boltzmann (PB) framework and another within the framework of size-modified PB theory, accounting for the finite-size effects. 
The construction of these models relies on  a new diffuse interface definition $u(x)$, which represents  the probability of a point  $x $  found in the solute phase among all microstates.
Drawing upon principles of statistical mechanics and geometric measure theory, we rigorously demonstrate that the proposed models effectively capture EASE during the solvation process. 
Moreover, preliminary analyses indicate that the size-modified EASE functional surpasses its counterpart based on classic PB theory across various analytic aspects.
Our work is the first step towards calculating EASE   through the utilization of diffuse-interface VISM.
\end{abstract}


\maketitle

\section{Introduction}\label{Section: introduction}

In the quantitative analysis of biological processes, the complex  interactions between the solute and solvent  are typically described by solvation energies (or closely related quantities):
the free energy of transferring the solute (e.g. biomolecules, such as proteins, DNA, RNA) from the vacuum to a solvent environment of interest (e.g.  water at a certain ionic strength). 
There are two major approaches for solvation energy calculations: explicit solvent models and implicit solvent models \cite{doi:10.1021/ct400065j}.
Explicit models, treating both the solute and the solvent as individual molecules, are too computationally expensive for large  solute-solvent systems, such as the solvation of macromolecules in ionic environments.
In contrast, implicit models, by averaging the effect of solvent phase as continuum media \cite{Feig:2004b,Baker:2005,Boschitsch:2004, Baker10037, BOTELLOSMITH2013274, doi:10.1021/jp7101012, https://doi.org/10.1002/jcc.23033}, are much more efficient and thus are able to handle much larger systems \cite{Baker10037, doi:10.1063/1.4745084,Grochowski:2007,Lamm:2003,Fogolari:2002,Tjong:2007b,Mongan:2007,Grant:2007}.

An inevitable prerequisite for describing the solvation energy in implicit solvent models is an interface separating the discrete solute and the continuum solvent domains.
All of the physical properties related to solvation processes, including biomolecular surface area, biomolecular cavitation volume, p$K_a$ value and  electrostatic free energy, are very sensitive to the interface definition \cite{Dong:2006,Swanson:2005a,Wagoner:2006}. 
There are a number of different surface definitions, which include the van der Waals  surface, the solvent excluded surface   and the solvent accessible surface.
These surface definitions have found many successful applications in biomolecular modeling \cite{Spolar:1989,Crowley:2005,Kuhn:1992,Dragan:2004}.
However, these predetermined interfaces are {\em ad hoc} partitions and thus either non-negligibly overestimate or underestimate the solvation free energies\cite{Wagoner:2006}. 
Moreover, none of them takes into account the minimization of interfacial free energies during  the  equilibrium solvation.

Variational implicit solvation models (VISM) stand out as a successful approach to compute the disposition of an interface separating the  solute and the solvent \cite{Wei:2005,Bates:2008,Cheng:2007e,CChen:2009, MR2719171, Zhou2020curvature,MR3740372}.
In a VISM, the desired   interface profile is obtained by optimizing a solvation energy functional coupling 
the discrete description of the solute and the continuum description of the solvent in addition to polar and non-polar interactions.

However, traditional sharp-interface VISMs have limitations in capturing the inherent randomness of the solute-solvent boundary. This randomness arises from factors such as atomic vibrations and thermodynamic fluctuations. In practical applications, ions, solutes, and solvent molecules are not rigid entities; they undergo small or significant conformational changes. 
The disposition of the solute-solvent interface is influenced by both the structure of the solute molecule and the surrounding solvent configuration.
Due to thermodynamic  fluctuations, when a fixed  solute molecular  structure being considered, solvent molecules and ions   can still adopt   different configurations \cite{Ball2003}, equivalent to forming different microstates in the language of statistical mechanics.
This highlights that the position of the solute-solvent interface is not unique, and a single, fixed configuration does not capture the full range of possible solute-solvent interactions.

On the other hand, experimentally observable quantities are ensemble averaged.
Utilizing the energy computed from a single microstate to predict the averaged energy from all microstates is inherently prone to inaccuracy.
Indeed, studies show that disregarding thermodynamic  fluctuations during the solvation process  can cause severe errors in predictions
of solvation free energies  \cite{doi:10.1021/jp0764384}.
Therefore, it is of imminent practical importance to develop a solvation model capable of calculating
ensemble average solvation energy (EASE) with thermodynamic fluctuations being taken into account.

According to statistical mechanics,  the ensemble  average of a quantity is the  weighted average (by means of the Boltzmann weight) of its values  among all microstates.
However, in practice, it is a challenging task to directly compute the EASE of a solute-solvent system by means of this definition.
If one only samples a handful of random realizations, it is likely that one can  only  capture a outlaw behavior rather than the meaningful average behavior of the stochastic regime.
In contrast,  a tremendous sampling of random realizations, although delivers a more accurate approximation of EASE, is computationally too expensive, sometimes even unaffortable.

To address this dilemma, we provide an alternative approach to the computation of EASE in Section~\ref{Section: ensemble average solvation functional}.
More precisely, we show that instead of computing and averaging the solvation energies among all microstates, one should characterize the ``mean behavior" of the solute-solvent interfaces.
This gives rise to a novel  ``interface" profile  $u:\Omega\to [0,1]$ such  that $u(x)$ represents the probability of a point  $x\in \Omega$  found in the solute phase among all microstates.

Based on this insight, we rigorously demonstrate that the EASE of a solute-solvent system can be effectively captured by using a  VISM  defined in terms of this ``interface" profile $u$. This modeling paradigm is highly flexible, enabling us to incorporate various considerations into EASE modeling.

As an illustrative example, we develop two EASE functionals, one within the framework of classic Poisson-Boltzmann (PB) theory and the other within the size-modified PB theory framework to account for finite size effects. These proposed models have the potential to significantly expedite the computation of EASE. They achieve this by using a single diffuse-interface profile instead of numerous sharp interfaces in various microstates, making the computation more efficient and computationally accessible.


The proposed models introduced in this work fall within the category of diffuse-interface VISM     \cite{MR3740372, MR3396402, Zhao_2013_Phase, MR2719171,PhysRevLett.96.087802, Wei2010, Wei_2016_differential,Wang_2015_Parameter,Zhao_2011_Pseudo}.
This family of models replaces the traditional sharp solute-solvent interface with a diffuse-interface profile, denoted as $u:\Omega\to \R$. 
Notably, our EASE models, as described in \eqref{ensemble average total energy type II} and \eqref{ensemble average solvation energy-steric}, have strong connections with one of the most widely employed diffuse-interface VISMs, the Geometric Flow-Based VISM (GFBVISM)  \cite{MR2719171,PhysRevLett.96.087802, Wei2010, Wei_2016_differential,Wang_2015_Parameter,Zhao_2011_Pseudo}.

Furthermore, our work can be viewed as an enhancement of GFBVISM in several crucial aspects.

First, our research provides a rigorous mathematical foundation for the physical interpretation of the diffuse-interface profile $u$ and the energy predictions generated by our models and GFBVISM. This clarification establishes a link between the sharp-interface and diffuse-interface VISM models, demonstrating that the diffuse-interface model indeed computes ``mean" energies consistent with sharp-interface models.

Secondly, we have incorporated two  physical constraints into the formulations of \eqref{ensemble average total energy type II} and \eqref{ensemble average solvation energy-steric}. These constraints effectively  eliminate the ill-posed issues associated with GFBVISM and guarantee  the optimal diffuse-interface profiles align with the physical principles.

Lastly, our first model,  \eqref{ensemble average total energy type II}, corrects the formulation of the polar component of GFBVISM, while the second model,  \eqref{ensemble average solvation energy-steric}, builds upon this correction by incorporating finite size effects into the model. This results in a more accurate representation of the solvation process, which is a significant improvement over the original GFBVISM formulation.

The rest of this paper is organized as follows. 
Section~\ref{Section: Solvation model} provides an introduction to one of the most widely used formulations of sharp-interface VISMs.
In Section~\ref{Section: ensemble average solvation functional}, we present the development of two Ensemble Average Solvation Energy (EASE) functionals within the frameworks of classic Poisson-Boltzmann theory and the size-modified Poisson-Boltzmann theory, taking into account finite size effects.
Section~\ref{Secion: analysis of ensemble average energy} comprises preliminary analyses of the proposed EASE models, with a particular emphasis on the size-modified EASE functional.
In Appendix~\ref{Appendix A}, we gather relevant information and facts about a class of strictly convex functionals related to the polar portion of EASE.
Appendix~\ref{Appendix B} offers a brief comparative analysis between our newly proposed models and the Geometric Flow-Based Variational Implicit Solvent Model (GFBVISM).
Finally, in Section~\ref{Section: Conclusion}, we draw conclusions and summarize the key findings and implications of our research.



\section{Ensemble Averaged Solvation Energy}\label{Section: Solvation model}

\noindent
{\bf List of Notations:} 
Given any open sets $U$ and $\Omega$,
$U\Subset \Omega$ means that $\overline{U}\subset \Omega$ and $\overline{U}$ stands for the closure of $U$.
We denote by $\cL^N$ and $\cH^{N-1}$ the $N-$dimensional Lebesgue measure and the $(N-1)-$dimensional Hausdorff measure, respectively.
For any $1\leq p \leq \infty$,   $p' $ is the H\"older conjugate of $p$. 
The phrase l.s.c is the abbreviation of lower semi-continuous.
$\chi_E$ always denotes the characteristic function of a set $E$.
\\
For any two Banach spaces $X$ and $Y$, the notation $\cL(X,Y)$ stands for the set of all bounded linear operators from $X$ to $Y$ and $\cL(X):=\cL(X,X)$. 


\subsection{Background}\label{Section: background}
We consider a solute-solvent system with a fixed biomolecular structure contained in a    bounded Lipschitz domain $\Omega \subset \R^3$ by using a (statistical) grand canonical ensemble.
The temperature, chemical potentials and volume of the system are kept constant.
Suppose that   there are $N_c$ ion species   in $\Omega$ and the system contains $N_c +2$ types of particles, i.e. the solute and solvent molecules and $N_c$ ion species.
We assume that the solute atoms are  centered at   $x_1, \cdots ,  x_{N_m} \in \Omega$.

\noindent
Define a probability space $(\cS, \cF,\cP)$, where 
the sample	 space $\cS=\{\cS_\alpha: \alpha\in A \}$ denotes the set of   all possible states of the system with $A$ being the index set of the states. $\cF$ is a $\sigma$-algebra of $\cS$ and $\cP$ is the probability measure.
Because a biomolecular structure is fixed, each state $\cS_\alpha$   is uniquely determined by a function
\begin{equation}
\label{fn cC}
\cC_\alpha\in L^1(\Omega; \R^{N_c+1}) : \, \,  \cC_\alpha(x)= (\cC_{\alpha,0}(x), \cC_{\alpha,1}(x), \cdots, \cC_{\alpha, N_c }(x)),
\end{equation}
where   $\cC_{\alpha,j}(x)$ is the (number) concentration of  the   $j$-th ion species at $x$ for $j=1,\cdots, N_c$, or the (number) concentration of solvent molecules for $j=0$.
In state $\cS_\alpha$, a   point $x\in \Omega$ is    in the solute phase   if $\cC_\alpha(x)=(0,\cdots,0)$; otherwise in the solvent phase.
A point that is not occupied by any particle is typically considered to be in the vacuum. However, due to the similar dielectric constants of the solute and the vacuum, it is reasonable to assume that such points can be treated as being in the solute phase.

\noindent
It is important to note that different states can result in the same solute and solvent phases. For instance, consider the scenario where a solvent molecule and an ion are interchanged, which are originally located at different positions, within a given state.
Allowing for a slight abuse of notation, we refer to the subset of states within $\cS$ that share the same solute and solvent phases as a {\em microstate} of the solute-solvent system.
Assume that the system undergoes $K$ microstates: $ \cM_1,\cdots, \cM_K \in \cF$; and each $\cM_k$ occurs with a probability $p_k= \cP(\cM_k)$.
For notational brevity, we put $\cK=\{ 1,2,\cdots,K\}$. 
This leads to a decomposition    $\cS=\bigsqcup\limits_{k\in \cK} \cM_k$; and thus $\sum_{k\in \cK}p_k=1$. 

\noindent
We define several  random variables that will be extensively used in this article.
\begin{itemize}
\item $X:\cS\to \cK$, where $X=\sum_{k\in \cK} k \chi_{\cM_k}$. Thus $X(\cS_\alpha)$ indicates the microstate that   $\cS_\alpha$ belongs to.
\item $\cC_j : \cS\to  L^1(\Omega)$, where $\cC_j( \cS_\alpha) =\cC_{\alpha,j} $ with $j\in \{0,1,\cdots,N_c\}$. Thus when $j\in \{1,\cdots,N_c\}$, $\cC_j(\cS_\alpha)(x)$   is the  (number) concentration   of the $j$-th ion species at   $x\in \Omega $ in   $\cS_\alpha$; $\cC_0(\cS_\alpha)(x)$   is the  (number) concentration   of the  solvent molecule at   $x\in \Omega $ in   $\cS_\alpha$. 
\end{itemize}
Let $Z$ be a Banach space. The ensemble average of a random variable $Y:\cS \to Z$ (in $\cS$)  is defined as
\begin{equation}
\label{def ensemble average}
\langle Y \rangle := \bE(Y) = \bE[ \bE(Y|X=k)]=\sum_{k\in \cK} p_k  \bE(Y|X=k).
\end{equation}
Similarly, if one considers $\cM_k$ as a (statistic) ensemble, the ensemble average of $Y$ in $\cM_k$ is defined by
\begin{equation}
\label{def ensemble average microstate}
\langle Y \rangle_k := \bE(Y| X=k)  .
\end{equation}
Throughout the remainder of this article, unless explicitly stated otherwise, the ensemble average of a quantity $Y$ always refers to the object defined by \eqref{def ensemble average}.


\subsection{A Sharp Interface VISM}\label{Section: Sharp interface}

In this section, we will introduce the solvation energy defined in a microstate $\cM_k$ with a fixed sharp solute-solvent interface.
Assume that in  microstate $\cM_k$, the solute phase is represented by a Caccioppoli subset $D_k \Subset \Omega$, or equivalently, $U_k=\Omega\setminus \overline{D}_k$ is the solvent phase.
Then the solute-solvent interface  is given by $ \partial D_k$.
See Figure~\ref{fig:sharp}(A) for an illustration.
\begin{figure}[hbt!]
\begin{subfigure}{.5\textwidth}
  \centering
  \includegraphics[scale=.24]{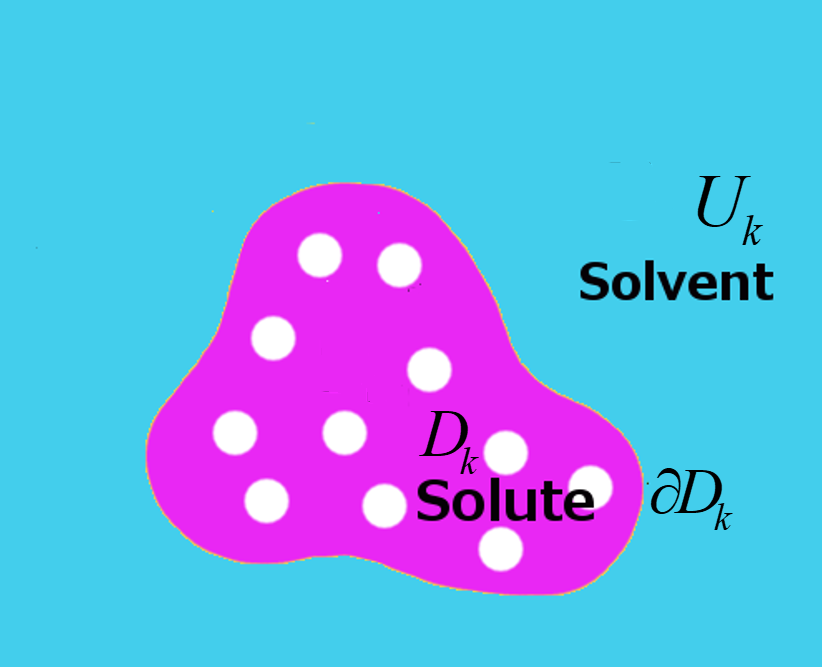}
   \caption{}
\end{subfigure}%
\begin{subfigure}{.5\textwidth}
  \centering
\includegraphics[scale=.4]{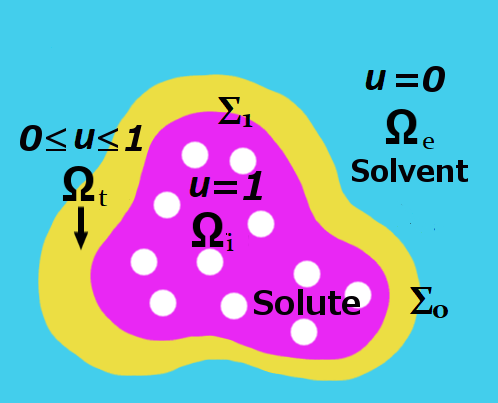}
  \caption{}
\end{subfigure}
\caption{\Small (A)  Illustration of a microstate $\cM_k$: $D_k$ is the solute   region; $U_k$ is the solvent  region; $\partial D_k$ is the solute-solvent interface;   (B) Domain decomposition for a grand canonical ensemble: $\Oi$ is the region   occupied by the solute in all microstates; $\Oe$ denotes the region   occupied by the solvent in all microstates; $\Ot$ denotes the transition region where $0\leq u \leq 1$.}
\label{fig:sharp}
\end{figure}

In this article,   solute atoms are treated as hard spheres.
More precisely, we consider the atom centered at $x_i$, $i=1,\cdots,N_m$, as a sphere  of radius $\sigma_i>0$, i.e., $B(x_i,\sigma_i)$, for which 
 solvent molecules and  ions cannot enter. 
On the other hand,  the solvent and ion concentrations are the same as their bulk concentrations in regions sufficiently far away from the solute.
As a consequence, such regions cannot be contained in the solute phase in any microstate. 
Based on these observations, we may assume that there are two open  subsets, $\Oi$ and $\Oe$, of $\Omega$ with $ \bigcup_{i=1}^{N_m} \overline{B}(x_i,\sigma_i) \subset  \Oi \Subset   \Omega\setminus \overline{\Omega}_{\rm e}$   and $\partial\Omega \subset \partial\Oe$ such that
\begin{equation}
\label{Constrain 0}
\Oi  \subset D_k \subset \Omega\setminus \overline{\Omega}_{\rm e}   \quad \text{for all }k\in \cK.
\end{equation}
For instance, one can take   $\Oi$ to be the region enclosed by a smoothed solvent excluded surface and $\Oe$ to be region outside a smoothed solvent accessible surface.
Let  $\Ot:=\Omega\setminus \left(\overline{\Omega}_{\rm i}\cup \overline{\Omega}_{\rm e} \right)$ be the transition region.
By Assumption~\eqref{Constrain 0}, in all microstates, 
\begin{itemize}
\item ions are located outside $\Oi$, and
\item the solute-solvent interfaces are located inside $\overline{\Omega}_{\rm t}$.
\end{itemize}
In addition, we define $\Sigma_1=\partial\Oi$ and $\Sigma_0=\partial\Oe\setminus \partial\Omega$.
We assume that $\Sigma_0, \Sigma_1$ are $C^2$ and $\Ot$ is connected so that $\Sigma_0\cap \Sigma_1 = \emptyset$. 
See Figure~\ref{fig:sharp}(B) for an illustration.

The hydrophobicity of the amino acids in the biomolecule varies from position to position.
To account for this phenomenon, we introduce a positive variable surface tension function $\theta \in C^1(\overline{\Omega}_{\rm t})$.
Without loss of generality, we may extend $\theta$ to be a   function in $C^1(\overline{\Omega})$, still denoted by $\theta$, such that
\begin{align}
\label{bound of theta}
 \theta_0 \leq \theta(x) \leq \theta_1 ,\quad x\in \Omega 
\end{align}
for some constants $0<\theta_0<\theta_1$.
By \cite[Lemma~1]{MR2345913},
\begin{equation}
\label{Def: weighted total variation}
\int_\Omega \theta(x) d| D f(x)|=\sup \left\{ \int_\Omega f(x) \nabla \cdot (\theta(x) \phi(x))\, dx : \,  \phi\in C^1_c(\Omega;\R^3), \, \|\phi\|_\infty \leq 1  \right\}, \quad f\in BV(\Omega).
\end{equation}
See  \cite{MR2719171} for the definition of $BV$-functions.
It is clear that, for any $f\in BV(\Omega)$ with $f\equiv 0$ in $\Oe$ and $f\equiv 1$ in $\Oi$, the value of
$ 
\int_\Omega \theta(x) d |Df(x)|   
$ 
is independent of how $\theta$ is extended outside $\overline{\Omega}_{\rm t}$.

As proposed in \cite{PhysRevLett.96.087802, doi:10.1063/1.2171192},  
the solvation free energy in microstate $\cM_k$ predicted by a sharp-interface VISM  is   the value of the following energy functional   evaluated at $\chi_{D_k}$:
\begin{gather}
\cE(\chi_{D_k}) =\Inp(\chi_{D_k})  + \Ip(\chi_{D_k},\psi) ,
\label{Sharp interface energy functional}
\end{gather}
where  $\psi:\Omega\to \R$ is the electrostatic potential.
The two components $\Inp$ and $\Ip$ are termed the nonpolar and polar portion of the solvation energy, respectively. 
For every fixed $D_k$, $\psi=\psi_{\chi_{D_k}}$ is chosen to maximize $\Ip(\chi_{D_k},\psi)$ among all $\psi$ taking a predetermined Dirichlet boundary value $\psi_D$ on $\partial\Omega$.

The nonpolar solvation energy consists of three parts:
\begin{equation}
\label{Def: nonpolar energy}
\Inp(\chi_{D_k}) =  \int_\Omega   \theta d|D\chi_{D_k}|   + P_h {\rm Vol}(D_k)   +\int_{U_k}   \rho_s U^{\mathrm{vdW}}(x)    \, dx .
\end{equation} 
When $\theta\equiv \gamma$ for some constant $\gamma>0$, the first term reduces to 
$ 
\gamma {\rm Per}(D_k;\Omega)
$ 
with ${\rm Per}(D_k;\Omega)$ being the perimeter (in $\Omega$) of the biomolecule region $D_k$. This term  measures the disruption of intermolecular and/or intramolecular bonds that
occurs when a surface is created.   
In addition, ${\rm Vol}(D_k)$ represents the volume of   $D_k$; 
$P_h$ is the (constant) hydrodynamic pressure.
Therefore, $   P_h {\rm Vol}(D_k)  $ is the mechanical work of creating the biomolecular size vacuum
in the solvent.
In the last integral,
$\rho_s$ is the (constant) solvent bulk density; and $U^{\mathrm{vdW}} $ represents the   Lennard-Jones potential \cite{Wagoner:2006}; as such $U^{\mathrm{vdW}}\in C^\infty(\overline{\Omega}\setminus \{x_1,\cdots,x_{N_m}\})$.

Currently, one of the most widely-used   polar solvation models   is based on the  Poisson-Boltzmann (PB) theory \cite{Grochowski:2007,Lamm:2003,Fogolari:2002,FOGOLARI1997135,Baker:2005}.
In the framework of classical PB theory, the polar energy is expressed as 
\begin{align}
\label{Def: polar energy}
\begin{split}
\hspace*{-2em}
\Ip(\chi_{D_k},\psi)
=     \int_{D_k}   \left[    \rho (x)  \psi(x)     - \frac{\epsilon_p}{8\pi} |\nabla \psi (x) |^2 \right]\, d x   
    - \int_{U_k} \left[ \frac{\epsilon_s}{8\pi} |\nabla \psi (x) |^2  +   \beta^{-1} \sum\limits_{j=1}^{N_c} c_j^\infty  (e^{- \beta q_j \psi(x)    }-1)         \right]\, d x,
\end{split}
\end{align} 
where $\rho$ is an $L^\infty$-approximation of the solute partial charges    supported  in $\bigcup_{i=1}^{N_m} B(x_i,\sigma_i)$.
$\epsilon_p$ and $\epsilon_s$ are the dielectric constants of the  solute  and the solvent, respectively. 
Usually, $\epsilon_p\approx 1$ for the protein and $\epsilon_s\approx 80$ for the water.
$q_j$ is the
charge of ion species $j$, $j=1,2,\cdots,N_c$, and $\beta=1/k_B T$, where $k_B$ is the Boltzmann constant, $T$ is the (constant) absolute  temperature.
$c_j^\infty $ is the (constant) bulk concentration of the $j-$th ionic species. 
For notational brevity, we define
\begin{equation}\label{Def fn B}
B(s)= \beta^{-1} \left[ \sum\limits_{j=1}^{N_c} c_j^\infty  \left( e^{- \beta s q_j } -1 \right) \right] ,\quad s\in \R.
\end{equation}
In addition, we assume   the charge neutrality condition  
\begin{equation}\label{neutral}
\sum\limits_{j=1}^{N_c} c_j^\infty q_j=0.
\end{equation}
It is important to observe that $B(0)=0$  and, by \eqref{neutral}, $B^\prime(0)=0$ and $B^\prime(\pm \infty)=\pm \infty$. Further,  $B^{\prime\prime}(s)>0$. We thus conclude that $B(0)=\underset{{s\in \R}}{\min} B(s)$ and $B$ is strictly convex.

To sum up, in microstate $\cM_k$, the solvation energy is given by 
\begin{equation}
\begin{split}\label{sharp-interface energy functional}
\cE(\chi_{D_K} )= & \int_\Omega   \theta d|D\chi_{D_k}|   + P_h {\rm Vol}(D_k)     +\int_{U_k}   \rho_s U^{\mathrm{vdW}}(x)    \, dx  + \int_{D_k}   \left[    \rho (x)  \psi(x)     - \frac{\epsilon_p}{8 \pi} |\nabla \psi (x) |^2 \right]\, d x   \\
&    - \int_{U_k} \left[ \frac{\epsilon_s}{8\pi} |\nabla \psi (x) |^2  +   B(\psi(x))       \right]\, d x,
\end{split}
\end{equation}
where $\psi=\psi_{\chi_{D_k}}$  solves   the PB equation:
\begin{equation}
\label{PBE}
\nabla \cdot ( (\epsilon_p \chi_{D_k} + \epsilon_s \chi_{U_k} ) \nabla \psi)  -4\pi \chi_{U_k} B'(\psi)+ 4\pi \rho=0   \quad \text{in }\Omega
\end{equation}
in  the admissible set
\begin{equation}
\label{Def cA}
\cA =\{\psi\in H^1(\Omega): \, \psi|_{\partial \Omega}=\psi_D\} \quad \text{for some } \psi_D\in W^{1,\infty}(\Omega).
\end{equation}

\begin{remark}\label{Rmk: PBE}
Now, we will present a concise discussion regarding the expression for polar energy expression \eqref{Def: polar energy}.
This discussion  will   serve as a crucial foundation for the derivation of EASE.
To begin, let us review the definitions of the random variables $X$ and $\cC_j$ defined in Section~\ref{Section: background}.

\noindent
A pivotal element of the Poisson-Boltzmann theory is an electrostatic potential  $\psi$, which is identical across all states within a fixed (statistical) ensemble under consideration. 
Consider a fixed microstate $\cM_k$ as a (statistical) ensemble. In $\cM_k$, the electrostatic potential  $\psi$ satisfies the fundamental equation of electrostatics, the Poisson  equation:
\begin{equation}
\label{Poisson eq}
-\nabla \cdot [ (\epsilon \chi_{D_k} + \epsilon_s \chi_{U_k}) \nabla \psi ]  = 4\pi \left( \rho + \sum\limits_{j=1}^{N_c} q_j c_j^k  \right) ,
\end{equation}
where  $c_j^k$ is the  (number) concentration  of the $j$-th ion speices  in  $\cM_k$, i.e.
\begin{equation}
\label{ion concentration in k}
c_j^k = \langle \cC_j \rangle_k =\bE(\cC_j |X=k);
\end{equation}
and thus $c_j^k\equiv 0$  in $D_k$. 
However, since   the probability distribution of $\cC_j $ is unknown, it is impossible to directly compute $c_j^k$ by means of \eqref{ion concentration in k}. 
Indeed,   the expression of   $c_j^k$ should be derived from the Helmholtz free energy, cf. \cite{FOGOLARI1997135, doi:10.1137/080712350, MR4218921}.

To find the Helmholtz free energy $H_k$ in  $\cM_k$, we consider the random variable $E_{\psi}:\cS\to \R$,  where
\begin{align}\label{internal energy rv}
 E_{\psi}(\cS_\alpha) 
=   \int_{D_k}    \left(    \rho \psi   - \frac{\epsilon_p}{8\pi} |\nabla \psi|^2  \right)  dx   
  +  \int_{U_k}  \left[ \sum\limits_{j=1}^{N_c} \left(   q_j \cC_{\alpha,j} \psi  -  \cC_{\alpha,j} \mu_j^\infty \right) - \frac{\epsilon_s}{8 \pi} |\nabla \psi|^2  \right] dx 
\end{align}
is the internal energy in   state $\cS_\alpha$ with $\mu_j^\infty$ being the chemical potential of the $j$-th ion species, which is uniform everywhere in a grand canonical ensemble. 
Then the ensemble average  of $E_\psi$ in $\cM_k$ is given by
\begin{align}\label{Def: internal energy}
 \langle E_\psi \rangle_k =\bE( E_\psi |X=k )= \int_{D_k}   \left(     \rho \psi    - \frac{\epsilon_p}{8\pi} |\nabla \psi|^2 \right)  dx   
  +  \int_{U_k}  \left[ \sum\limits_{j=1}^{N_c} \left(   q_j c_j^k \psi  -  c_j^k \mu_j^\infty \right) - \frac{\epsilon_s}{8 \pi} |\nabla \psi|^2  \right] dx .
\end{align}
Meanwhile, within the framework of classic PB theory, the entropy $S_k$ in $\cM_k$ is given by
\begin{align}\label{Def: entropy}
  -T S_k  
=  \beta^{-1}\int_{U_k}   \sum\limits_{j=1}^{N_c}c_j^k \left[  \ln(  \sv c_j^k ) -1  \right]\, dx,
\end{align}
where $1/\sv$ is a reference concentration. 
It is crucial to highlight that, in statistic mechanics,   entropy is not a random variable of $(\cS,\cF,\cP)$ and  should not be obtained by ensemble averaging its counterparts in all states.

\noindent
The Helmholtz free energy $H_k$  is given by
$H_k = \langle E_\psi \rangle_k  - T S_k$.
Equating the first variation of $H_k$ with respect to $c_j^k$ to zero gives  
$$
\mu_j^\infty=  \beta^{-1} \ln(\sv c_j^k (x) ) + q_j \psi(x)     \quad  \text{in } U_k .
$$
This leads to the relation $c_j^k(x)= \chi_{U_k}(x) \sv^{-1} e^{\beta\mu_j^\infty} e^{-\beta q_j \psi(x)} $. 
When $x$ is sufficient far away from the solute, we have $\psi(x)=0$ and $c_j^k(x)=c_j^\infty$.
This implies that $c_j^\infty= \sv^{-1} e^{\beta\mu_j^\infty}$.
Hence we obtain   the relation
 \begin{equation}
\label{ion concentration and Boltzmann distribution}
c_j^k(x)= \chi_{U_k}(x) c_j^\infty e^{-\beta q_j \psi(x)} .
\end{equation}
This is exactly the Boltzmann distribution.
It is important to emphasize that the term $\chi_{U_k}$ in Equation \eqref{ion concentration and Boltzmann distribution} arises from the fact that, in $H_k$, the domain of integration for any integral involving $c_j^k$ is restricted to $U_k$. This restriction is necessary because $c_j^k\equiv 0$ within the domain $D_k$. By using the characteristic function $\chi_{U_k}$, we ensure that the integration is performed only over the relevant region $U_k$ where $c_j^k$ is nonzero.
Plugging \eqref{ion concentration and Boltzmann distribution} back into the expression of $H_k$ yields 
\begin{align}
\label{Helmholtz energy}
\begin{split}
H_k
=     \int_{D_k}   \left(     \rho (x)  \psi(x)     - \frac{\epsilon_p}{8\pi} |\nabla \psi (x) |^2 \right)\, d x   
    - \int_{U_k} \left( \frac{\epsilon_s}{8\pi} |\nabla \psi (x) |^2  +   \beta^{-1} \sum\limits_{j=1}^{N_c} c_j^\infty   e^{- \beta q_j \psi(x)    }          \right)\, d x.
\end{split}
\end{align}
Since the polar energy equals zero when $\psi$ vanishes everywhere, 
following \cite{Sharp1990}, a constant term
\begin{equation}
\label{reference state}
\beta^{-1} \int_{U_k}     \sum\limits_{j=1}^{N_c}   c_j^\infty    dx
\end{equation}
should be added to \eqref{Helmholtz energy}, which yields the polar energy expression~\eqref{Def: polar energy} in   $\cM_k$.
Note that  $\sum\limits_{j=1}^{N_c} q_j c_j^k = -\chi_{U_k} B'(\psi) $, see \eqref{Def fn B}. 
Plugging this expression into \eqref{Poisson eq} shows that $\psi$  solves \eqref{PBE}.
By Lemma~\ref{Lem: GPBE estimate},  $\psi$ maximizes $\Ip(\chi_{D_k},\cdot)$ in $\cA$.

To summarize, within the framework of the classic Poisson-Boltzmann theory, the following points are crucial in the derivation of EASE in a fixed ensemble:

\begin{itemize}
\item 
Universal Electrostatic Potential: It is necessary to use a universal electrostatic potential, denoted as $\psi$, which applies across all microstates, when deriving the Boltzmann distribution. This potential is independent of specific microstates and is used consistently in the calculations.
\item
Non-Random Nature of Entropy: In the context of EASE, it is important to note that one component of the solvation energy, namely entropy, is not considered a random variable of the system $(\cS,\cF,\cP)$. Therefore, when calculating the entropy within EASE, it should not be obtained by averaging its counterpart in the individual microstates. 
\end{itemize}
\noindent
By considering these factors, EASE can be derived within a fixed ensemble, incorporating an appropriate electrostatic potential and accounting for the non-random nature of entropy.
\end{remark}

\section{Modeling of Ensemble Average  Solvation Energy}\label{Section: ensemble average solvation functional}


If one  tries to directly use  \eqref{def ensemble average}  and compute EASE,  an obvious barrier is how to calculate the probability $p_k$.
Although, according to statistical mechanics,  $p_k$ can be obtained by means of the Boltzmann weight,  an accurate approximation of the Boltzmann weight    requires a sufficiently large sampling of   microstates, which creates a significant computational burden. 
Indeed, in the routine calculation of the  EASE, one needs to carry out certain explicit solvent simulations, e.g. molecular dynamics (MD), to obtain thousands of solute-solvent states and perform energy calculations for each state.

The core of our modeling paradigm is, instead of ensemble averaging the outputs of the solvation energy functional~\eqref{sharp-interface energy functional},  one should ensemble average the inputs, i.e. $\chi_{D_k}$, to obtain a diffuse-interface profile.
Then the EASE can be computed by using this one profile instead of thousand of solute-solvent states.
This leads to the following definition: 
\begin{equation}
\label{def of u}
u=\sum_{k\in\cK}p_k \chi_{D_k} .
\end{equation}
As such, $u(x)$ represents the probability of   $x\in \Omega$  found in the solute phase among all microstates.
Figure~\ref{fig:sharp}(B) shows the range of $u$ in different regions.
The definition of $u$ enforces  the following physical constraints
\begin{equation}
\label{constrain 1}
u(x)\in [0,1] \quad \text{for a.a. } x\in \Omega
\end{equation}
and
\begin{equation}
\label{constrain 2}
u=1 \quad \text{a.e. in }\Oi  \quad \text{and} \quad  u=0 \quad \text{a.e. in }\Oe .
\end{equation} 
Then the admissible set for $u$ is  defined as
\begin{equation}
\cX=\{   u\in BV(\Omega): \, u \, \text{ satisfies Constraints~\eqref{constrain 1}~and~\eqref{constrain 2}} \}  .
\end{equation}
To construct the EASE functional, we introduce the following  functions defined on $\cK$: 
\begin{equation}
\begin{split}
f_D: \cK \to BV(\Omega), \quad  & f_D(k)=\chi_{D_k} \\
f_U: \cK \to BV(\Omega), \quad  & f_U(k)=\chi_{U_k} \\
f_s: \cK \to \R  ,   \quad &   f_s(k)=  \int_\Omega \theta d |D \chi_{D_k}| \\
f_{\rm i} : \cK \to \cL(L^1(\Omega),\R) = L^\infty(\Omega) , \quad &   f_{\rm i}(k)v =\int_\Omega \chi_{D_k} v\, dx = \int_{D_k} v \, dx  , \quad v\in L^1(\Omega)\\
f_{\rm e} : \cK \to \cL(L^1(\Omega),\R) = L^\infty(\Omega) , \quad &   f_{\rm e}(k)v =\int_\Omega  \chi_{U_k}  v\, dx= \int_{U_k} v \, dx , \quad v\in L^1(\Omega).
\end{split}
\end{equation}
The ensemble averages of all bulk and interfacial energy components can be derived from the above functions. For example, 
\begin{itemize}
\item  the ensemble average of  interfacial energy, i.e. $\int_\Omega \theta d|D \chi_{D_k}|$, is given by $\langle f_s(X) \rangle$;
\item the ensemble average of the term $P_h {\rm Vol}(D_k)$ is given by $\langle f_{\rm i}(X) P_h \rangle$;
\item the ensemble average of the term $\int_{U_k} \rho_s U^{\rm vdW}\, dx$, is given by $\langle f_{\rm e}(X) (\rho_s U^{\rm vdW} ) \rangle$.
\end{itemize}


\begin{prop}
\label{Prop: bulk enegy}
Assume that $v\in L^1(\Omega)$   is an   energy density function. Then the ensemble average of the energy stored in the solute and solvent phases can be computed as follows:
$$
\langle f_{\rm i}(X) v \rangle  =    \int_\Omega u v  \, dx
\qquad 
\text{and}
\qquad
\langle f_{\rm e}(X) v \rangle  =   \int_\Omega (1- u) v \, dx.
$$
\end{prop}
\begin{proof}
The proof is straightforward. Indeed, we will only check the equality for $\langle f_{\rm i}(X) v \rangle$. Consider the random variable $f_{\rm i}(X) v : \cS \to \R$. Then
\begin{align*}
\langle f_{\rm i}(X) v \rangle= \bE[\bE(f_{\rm i}(X) v|X=k)]   =\sum\limits_{k\in \cK} p_k \int_{D_k}  v\, dx =  \int_\Omega \sum\limits_{k\in \cK} p_k \chi_{D_k} v\, dx  
=  \int_\Omega u v  \, dx.
\end{align*}
\end{proof}

Following the discussion in Remark~\ref{Rmk: PBE}, let $\psi:\Omega\to \R$ be
the universal electrostatic potential $\psi$, which is identical among all $\cS_\alpha$. To formulate the Poisson equation, we notice that the dielectric function    $ \epsilon_p f_D(k) + \epsilon_s f_U(k)=\epsilon_p \chi_{D_k} + \epsilon_s \chi_{U_k}$ in \eqref{Poisson eq} should be replaced by  the ensemble averaged one:
$$
\epsilon(u):=\bE(\epsilon_p f_D(X) + \epsilon_s f_U(X))=u  \epsilon_p +(1-u )\epsilon_s .
$$
Hence  $\psi \in \cA$  solves  
\begin{equation}
\label{Poisson eq grand}
-\nabla \cdot [  \epsilon(u) \nabla \psi ] = 4\pi\left(  \rho + \sum\limits_{j=1}^{N_c} q_j c_j    \right) \quad \text{in }\Omega  ,
\end{equation}
where  $c_j$ is the   ion concentration of the $j$-th ion species in the ensemble under consideration, i.e.
\begin{equation}
\label{mean ion con}
c_j =\bE(\cC_j)= \bE[\bE(\cC_j|X=k)]=\sum\limits_{k\in \cK} p_k c_j^k = \sum\limits_{k\in \cK} p_k c_j^k \chi_{U_k}. 
\end{equation}
Note that the function $c_j\equiv 0$ in the set $ \{u=1\}: = \{x\in \Omega: u(x)=1\} $.
By a similar argument to Remark~\ref{Rmk: PBE} and Proposition~\ref{Prop: bulk enegy},   the ensemble average of $E_\psi$, cf. \eqref{internal energy rv} and \eqref{Def: internal energy},  equals 
\begin{align}\label{Def: internal energy ensemble}
\notag \langle E_\psi \rangle &=  \sum_{k\in\cK} p_k \left[ \int_{D_k}   \left(     \rho \psi    - \frac{\epsilon_p}{8\pi} |\nabla \psi|^2 \right)  dx   
  + \int_{U_k} \left(   \sum\limits_{j=1}^{N_c} \left(q_j c_j^k \psi -c_j^k \mu_j^\infty \right)    -\frac{\epsilon_s }{8\pi}   |\nabla \psi|^2     \right)  dx  \right]\\
&=\int_\Omega \left[ \rho \psi + \sum\limits_{j=1}^{N_c} \left(q_j c_j \psi -c_j  \mu_j^\infty \right)    - \frac{\epsilon(u)}{8\pi}|\nabla \psi|^2    \right] \, dx.
\end{align}
In the following two subsections, we will derive two functionals for the EASE within the frameworks of classic PB theory and size-modified PB theory.

\subsection{Classic Poisson-Boltzmann Theory}\label{Section: classic PB}

In this subsection, we will first consider the classic PB theory, in which ions and solvent molecules are assumed to be point-like and have no correlation between their concentrations. 

\subsubsection{Ensemble Average Polar Energy}\label{Secion: ensemble average polar energy}

In the   classic PB theory framework, the entropy can be formulated as
\begin{align}\label{Def: entropy ensemble}
  -T S  
=  \beta^{-1}\int_{\{u< 1\}}   \sum\limits_{j=1}^{N_c}c_j  \left[  \ln(\sv c_j  ) -1  \right]\, dx.
\end{align}
The domain of integration $\{u<1\}=\{x\in \Omega: u(x) <1 \}$  is due to the fact that $c_j$ vanishes identically in $\{u=1\}$.
Applying a similar argument leading  to \eqref{ion concentration and Boltzmann distribution} to the Helmholtz energy $H=\langle E_\psi \rangle -ST$, one can obtain the Boltzmann distribution
\begin{equation}\label{express of cj}
c_j= \chi_{\{u<1\}}  c_j^\infty e^{-\beta q_j \psi }.
\end{equation}
After plugging \eqref{express of cj} into   $H$, as in \eqref{reference state}, the constant term
\begin{equation}
\label{reference state ensemble}
\beta^{-1} \int_\Omega \chi_{\{u<1\}}     \sum\limits_{j=1}^{N_c}   c_j^\infty    dx
\end{equation}
needs to be added to $H$ to adjust the reference state of the zero energy in the ensemble under consideration.
Finally, we arrive at  the   ensemble average polar energy
\begin{align}\label{ensemble average polar formulation 2}
\begin{split}
\Ip (u,\psi) = H+ \beta^{-1} \int_\Omega \chi_{\{u<1\}}     \sum\limits_{j=1}^{N_c}   c_j^\infty    dx= \int_\Omega   \left[    \rho \psi    - \frac{\epsilon(u) }{8\pi}|\nabla \psi|^2 - \chi_{\{u<1\}}  B(\psi)       \right]\, dx. 
\end{split}
\end{align}
On the other hand, replacing $c_j$ in \eqref{Poisson eq grand} by \eqref{express of cj} shows that $\psi\in \cA$ solves the generalized PB equation:
\begin{equation}
\label{GPBE2}
 \nabla \cdot [  \epsilon(u) \nabla \psi ] -4\pi \chi_{\{u<1\}}B'(\psi) + 4\pi \rho=0 \quad \text{in }\Omega.	
\end{equation}
Lemma~\ref{Lem: GPBE estimate} shows that $\psi$  maximizes  $\Ip (u,\cdot)$ in $\cA$ for any fixed $u\in \cX$.



\subsubsection{Ensemble Average  Interfacial Energies}\label{Section: ensemble average interfacial energy}
Following  the discussions in Section \ref{Secion: ensemble average polar energy}, the EASE in the framework of classic PB theory can be represented by 
\begin{align}\label{Def:ensemble average L-initial}
\langle  f_s(X)  \rangle  + \sL(u ) , \quad \text{where } \sL(u )=  \int_\Omega \left[ P_h u + \rho_s (1-u)  U^{\mathrm{vdW}}  \right]\, dx  + \Ip(u,\psi_u)  
\end{align}
with $\psi_u$   maximizing $\Ip(u,\cdot)$ in $\cA$.
It only remains to compute the ensemble averaged   interfacial energy 
$$
\langle  f_s(X)  \rangle = \bE[\bE(f_s(X)|X=k)]=  \sum_{k\in \cK} p_k \int_\Omega \theta d|D \chi_{D_k}|.
$$
Proposition~\ref{Prop: surface area} below shows that the integral $\int_\Omega \theta d|Du|$ can be used to approximate the ensemble averaged   interfacial energy.
Before proving Proposition~\ref{Prop: surface area}, we will need the following lemma.

\begin{lem}\label{Lem: modify sets}
Let $\{D_k\}_{k\in \cK}$
be a family of    Caccioppoli sets satisfying $\Oi  \subset D_k \subset \Omega\setminus \overline{\Omega}_{\rm e}  $  and $p_k\in [0,1]$ with $\sum_{k\in \cK} p_k=1$. 
Then for each $\varepsilon>0$,
there exists another  family $\{\widetilde{D}_k \}_{k\in \cK}$ of Caccioppoli sets  satisfying $\Oi  \subset \widetilde{D}_k \subset \Omega\setminus \overline{\Omega}_{\rm e}    $ and
\begin{align}
\label{EASP1}
\cH^2(\partial^* \widetilde{D}_k \cap \partial^* \widetilde{D}_j )=0 ,\quad \forall k,j\in \cK,\, \, k\neq j,
\end{align}
with $\partial^*D$ being the reduced boundary of a Caccioppoli set $D$.
Moreover,        
\begin{equation}\label{eq: modify energy}
\left| \sum\limits_{k\in \cK}p_k \int_\Omega \theta d |D \chi_{D_k}|  +   \sL (u  )  - \sum\limits_{k\in \cK}p_k \int_\Omega \theta d |D \chi_{\widetilde{D}_k}|   - \sL (\widetilde{u}  )     \right|<\varepsilon , 
\end{equation}
where 
$u=\sum\limits_{k\in \cK}p_k \chi_{D_k}$ and $\widetilde{u}=\sum\limits_{k\in \cK}p_k \chi_{\widetilde{D}_k}.$
\end{lem}
\begin{proof}
In view of  Lemma~\ref{Lem: GPBE converge}, it suffices to construct  a family of Caccioppoli sets $\{D_{k,n}\}_{n=1}^\infty$ satisfying Assumption~\eqref{EASP1} and 
$\Oi  \subset D_{k,n} \subset \Omega\setminus \overline{\Omega}_{\rm e}    $
such that  \begin{equation}
\label{Dkj}
\lim\limits_{n\to \infty} \| \chi_{D_k} - \chi_{D_{k,n}} \|_{L^1}=0 ,\quad \lim\limits_{n\to \infty}    \int_\Omega \theta d|D\chi_{D_{k,n}}| =  \int_\Omega \theta d|D\chi_{D_k } | .
\end{equation}
Indeed, letting $u_n= \sum_{k\in \cK}p_k \chi_{D_{k,n}}$, \eqref{Dkj} implies that $u_n\to u$ and $\chi_{\{u_n<1\}} \to \chi_{\{u<1\}}$ in $L^1(\Omega)$. Then it follows from Lemma~\ref{Lem: GPBE converge} and the dominated convergence theorem that $ \sL (u_n,\psi_{u_n} )\to  \sL (u,\psi_u )$ as $n\to \infty$.

Because $\Sigma_i$, $i=0,1$, are $C^2$, 
there exists some $\a >0$ such that $\Sigma_i$   has a tubular neighborhood $B_{\a}(\Sigma_i)$  of width $\a >0$,   cf. \cite[Exercise 2.11]{MR737190} and \cite[Remark 3.1]{MR4160135}. 
For sufficiently small $r\in (0,\a)$, let 
$$
D_k^r= (D_k \cup B_r(\Sigma_1) ) \setminus B_r(\Sigma_0).
$$
In view of the relation ${\rm Per}(S\cup T; \Omega)+ {\rm Per}(S\cap T; \Omega) \leq  {\rm Per}(S ; \Omega)+ {\rm Per}(T ; \Omega) $ for any Caccioppoli sets $S,T\subset \Omega$, $D_k^r$ are again Caccioppoli sets. 
Recall that ${\rm Per}(S;\Omega)$ is the perimeter of $S$ in $\Omega$.	
According to Lemma~\ref{tech lem 2}, we have
$$
\lim_{r\to 0^+} \left\| \chi_{D_k} - \chi_{D_k^r}  \right\|_{L^1} =0,  \quad \lim\limits_{r\to 0^+}  \int_\Omega \theta d |\chi_{D_k^r}| = \int_\Omega \theta d |\chi_{D_k} |  .
$$
Note that the proof of Lemma~\ref{tech lem 2} is independent of  other results in this article.
Therefore, without loss of generality, we may assume that  for some sufficiently small $r>0$
$$
\Oi \cup B(\Sigma_1, r)  \subset D_k \subset \Omega\setminus \left(\overline{\Omega}_{\rm e} \cup B_r(\Sigma_0) \right) .
$$

 \smallskip

\noindent
{\bf Claim 1:} 
For every $k$,  there exists a sequence $\{f_{k,j}\}_{j=1}^\infty \subset C^\infty(\overline{\Omega})$ such that  $0\leq f_{k,j} \leq 1$ a.e. in $\Omega$ and 
$$
\lim\limits_{j\to \infty} \| \chi_{D_k} - f_{k,j} \|_{L^1}=0 ,\quad \lim\limits_{j\to \infty}    \int_\Omega \theta d|Df_{k,j}| =  \int_\Omega \theta d|D\chi_{D_k } | .
$$
\begin{subproof}[Proof of Claim 1] 
We consider $\chi_{D_k}$ as an element in $BV(\R^3)$. 
Choose a sequence $\varepsilon_j \to 0^+$ and define
$f_{k,j}:=\eta_{\varepsilon_j} \ast \chi_{D_k}$, where $\eta_{\varepsilon_j}$ is the standard Friedrichs mollifying kernel. 
Then $\lim\limits_{j\to \infty} \| \chi_{D_k} - f_{k,j} \|_{L^1}=0$.

\noindent \cite[Corollary~1]{MR2345913} implies that 
$ 
\int_\Omega \theta d|D \chi_{D_k} (x)| \leq  \liminf\limits_{j\to \infty} \int_\Omega \theta d|D f_{k,j} (x)| .
$ 
Note that $f_{k,j}\in \cX \cap C^\infty(\overline{\Omega})$ for sufficiently large $j$.
For any $\phi\in C^1_0(\Omega;\R^3)$ with $\|\phi\|_\infty \leq 1 $, 
it follows from \cite[Lemma~1]{MR2345913} that
\begin{align*}
\int_\Omega f_{k,j} \nabla \cdot (\theta \phi)\, dx &= \int_\Omega \left(\eta_{\varepsilon_j}\ast \chi_{D_k} \right)\nabla \cdot (\theta \phi)\, dx =\int_\Omega \chi_{D_k} \nabla \cdot \left[ \eta_{\varepsilon_j} \ast (\theta \phi ) \right]\, dx  \\
&=  \int_\Omega \chi_{D_k} \nabla \cdot  \left[ \theta \left( \frac{ \eta_{\varepsilon_j} \ast (\theta \phi ) }{\theta} \right) \right]\, dx \\
& \leq \left\| ( \eta_{\varepsilon_j} \ast \theta )/\theta  \right\|_{L^\infty }  \int_\Omega \theta d |D \chi_{D_k} (x)|.
\end{align*}
Here  
$ 
( \eta_{\varepsilon_j} \ast \theta )(x) =\int_\Omega \eta_{\varepsilon_j} (x-y) \theta (y)\, dy$ for $ x\in \Omega.
$ 
Taking supremum over all such $\phi$, we derive that
$$
 \int_\Omega \theta d |D f_{k,j}(x)| \leq  \left\| ( \eta_{\varepsilon_j} \ast \theta )/\theta  \right\|_{L^\infty}  \int_\Omega \theta d |D \chi_{D_k}(x)|.
$$
By the uniform continuity of $\theta$, it is not a hard task to verify that
$ 
\lim\limits_{j\to \infty}\left\| ( \eta_{\varepsilon_j} \ast \theta )/\theta  \right\|_{L^\infty(\Omega )}=1.
$ 
Therefore, the above inequality implies that
$
\limsup\limits_{j\to \infty} \int_\Omega \theta d |D f_{k,j}(x)| \leq \int_\Omega \theta d |D \chi_{D_k}(x)|.
$
\end{subproof}

By the Sard's Theorem, there exists some $S \subset (0,1)$ with $\cL^1((0,1)\setminus S)=0$ such that, for all $t\in S$, the super-level set $E^t_{k,j}=\{ f_{k,j}>t\}$ has a smooth boundary. 
The coarea formula implies that 
\begin{align*}
\int_\Omega \theta d|D\chi_{D_k } |=  \lim\limits_{j\to \infty}    \int_\Omega \theta d|Df_{k,j}|= \lim\limits_{j\to \infty}\int_0^1   \int_\Omega \theta d|D\chi_{E^t_{k,j}}| \, dt   
\geq    \int_0^1  \liminf\limits_{j\to \infty}\int_\Omega \theta d|D\chi_{E^t_{k,j}}| \, dt.
\end{align*}
Therefore, for some $t\in   S$,
$$
\liminf\limits_{j\to \infty}\int_\Omega \theta d|D\chi_{E^t_{k,j}}|  \leq \int_\Omega \theta d|D\chi_{D_k } |.
$$
Pick the subsequence $\{j_n\}_{n=1}^\infty$ such that
$$
\liminf\limits_{j\to \infty}\int_\Omega \theta d|D\chi_{E^t_{k,j}}|=\lim\limits_{n\to \infty}\int_\Omega \theta d|D\chi_{E^t_{k,j_n}}|.
$$ 
On the other hand, we can infer from the Chebyshev's Theorem that
$$
\cL^3(E^t_{k,j_n}\setminus D_k)\leq \frac{1}{t}\| f_{k,j_n}-\chi_{D_k}\|_{L^1} ,\quad \cL^3(D_k \setminus E^t_{k,j_n} )\leq \frac{1}{1-t}\| f_{k,j_n}-\chi_{D_k}\|_{L^1}.
$$
This implies that 
$$
\lim\limits_{n\to \infty} \| \chi_{E^t_{k,j_n}} - \chi_{D_k}\|_{L^1} =0.
$$
Therefore, it follows from \cite[Corollary~1]{MR2345913} that $\widetilde{D}_{k,n}=E^t_{k,j_n}$ satisfies \eqref{Dkj} with $D_{k,n}$ replaced by $\widetilde{D}_{k,n}$.

Assume that $\cH^2(\partial \chi_{\widetilde{D}_{k ,n}} \cap \partial \chi_{\widetilde{D}_{j,n}})>0$ for some $k,j\in \cK$ and $k\neq j$.
Since $\partial \widetilde{D}_{k ,n}$ is $C^2$, 
there exists some $\a>0$ such that $\partial \widetilde{D}_{k ,n}$   has a tubular neighborhood $B_\a(\partial \widetilde{D}_{k ,n})$  of width $\a>0$.
Denote by $\nu$   the outward unit normal of $\widetilde{D}_{k ,n}$ pointing into $\Omega\setminus \widetilde{D}_{k ,n}$.
Then the map defined by
$$
\Lambda : \partial{\widetilde{D}_{k ,n}}\times (-\a,\a)\to B_\a(\partial \widetilde{D}_{k ,n}): \,  (x,r)\mapsto x+r \nu (x),
$$
is a $C^1$-diffeomorphism.
Let $\Gamma_r:=\Lambda(\partial{\widetilde{D}_{k ,n}},r)$. Then, for every $i\in \cK$, there are at most countably many $r\in (-\a,\a)$ such that $\cH^2(\Gamma_r \cap \partial \chi_{\widetilde{D}_{i,n}})>0$. 
Hence we can find $r\in (-\a,\a)$ sufficiently close to $0$ such that 
$$
\cH^2(\Gamma_r \cap \partial \chi_{\widetilde{D}_{i,n}})=0,\quad \forall i\in \cK
$$
and  $\Gamma_r \subset \Ot$ and
$$
\left\|\chi_{D_{k,n}} - \chi_{\widetilde{D}_{k,n}} \right\|_{L^1} + \left| \int_\Omega \theta d|D\chi_{D_{k,n}}| - \int_\Omega \theta d|D \chi_{\widetilde{D}_{k,n}}| \right| <1/n,
$$
where $D_{k,n}$  is the region enclosed by $\Gamma_r $.
Modifying all $\widetilde{D}_{k,n}$, $k\in \cK$, in such a way yields a family of smooth sets $\{D_{k,n} \}_{k\in \cK}$ satisfying $\Oi  \subset D_{k,n} \subset \Omega\setminus \overline{\Omega}_{\rm e}    $ and Assumptions~\eqref{EASP1} and \eqref{Dkj}.
\end{proof}

The lemma above elucidates that within a given ensemble, slight adjustments can be made to the solute regions $\{D_k\}_{k\in \cK}$ in order to attain \eqref{EASP1}, while ensuring that the alterations in the associated solvation energy remain ``infinitesimally" negligible.
Consequently, it is permissible to consistently presume that $\{D_k\}_{k\in \cK}$ adheres to \eqref{EASP1}.
\begin{prop}
\label{Prop: surface area}
Assume  that $\{D_k\}_{k\in \cK}$ satisfies \eqref{EASP1}.
Then 
$$
\int_\Omega \theta \, d  |D u | =\sum_{k\in\cK}p_k \int_\Omega \theta   \, d  | D \chi_{D_k} |.
$$
\end{prop}
\begin{proof}
By the De Giorgi's structure theorem, cf. \cite[Theorem~5.7.2]{MR1158660}, $\chi_{D_k}$ are 2-rectifiable, i.e. there exist Borel sets $F_k$ and $C^1$-functions $g_{k,n}: U_{k,n}\to \R^3$, $U_{k,n} \subset \R^2$ compact, such that $\|\partial D_k\|(F_k)=0$ and
$$
\partial^* D_k =F_k \cup \bigcup\limits_{n=1}^\infty g_{k,n}(U_{k,n}).
$$
Pick arbitrary $\varepsilon>0$.
For every $k\in \cK$, we can find $N_k=N_k(\varepsilon)$ such that
$$
\|\partial D_k\|(\Omega \cap \bigcup\limits_{n=N_k+1}^\infty g_{k,n}(U_{k,n}) )\leq \frac{\varepsilon}{K \theta_1 p_k}.
$$
Recall $K=|\cK|$.
Therefore, we obtain  compact sets 
$$
G_{k,\varepsilon}:= \bigcup\limits_{n=1}^{N_k(\varepsilon)} g_{k,n}(U_{k,n})  ,\quad  k\in \cK ,\quad \text{and}\quad G_\varepsilon:= \bigcup\limits_{k\in \cK}  G_{k,\varepsilon}.
$$
By the definition of the reduced boundary, the unit normal $\nu$ exists everywhere on $G_\varepsilon$. 
By the Urysohn's Lemma, there exists a continuous function $\phi:\R^3\to \R^3$ such that $\phi|_{G_\varepsilon}=\nu$. Restricting $\phi$ on $\overline{\Omega}$ and applying Stone-Weierstrass, we can find a smooth function $\phi_\varepsilon: \overline{\Omega}:\to \R^3$ such that  
$$
\|\phi_\varepsilon - \phi \|_{L^\infty(G_\varepsilon)} < \varepsilon.
$$
Then, in a neighborhood $H\Subset  \Omega$ of $G_\varepsilon$, we have  $|\phi_\varepsilon|>1/2$. Pick $h\in C^\infty_0(\Omega;[0,1])$ such   that $h\equiv 1$ in $H$. Setting 
$
\psi_\varepsilon(x)= h(x) \frac{\phi_\varepsilon(x)}{|\phi_\varepsilon(x)|},
$
it is an easy task to check that
$$
1\geq \psi_\varepsilon(x) \cdot \nu(x)\geq \frac{1-\varepsilon}{ 1+\varepsilon }, \quad x\in G_\varepsilon.
$$
By \cite[Lemma~1]{MR2345913}, we can  estimate
\begin{align*}
\int_\Omega \theta \, d|D u| \geq &  \int_\Omega   ( \theta \psi_\varepsilon ) \cdot dDu  =   \sum\limits_{k\in \cK} p_k \int_\Omega   \theta \psi_\varepsilon \cdot D \chi_{D_k}\, dx     
\geq   \sum\limits_{k\in \cK} p_k \int_{G_{k,\varepsilon}}   \theta \psi_\varepsilon \cdot D \chi_{D_k}\, dx     -  \varepsilon  \\
\geq & \sum\limits_{k\in \cK}   \frac{p_k (1-\varepsilon)}{ 1+\varepsilon }\int_{G_{k,\varepsilon}}   \theta \, d| D \chi_{D_k}| - \varepsilon \\
= & \sum\limits_{k\in \cK}   \frac{p_k (1-\varepsilon)}{ 1+\varepsilon } \left( \int_{\Omega}   \theta \, d| D \chi_{D_k}| -  \int_{\Omega \setminus G_{k,\varepsilon}}   \theta \, d| D \chi_{D_k}| \right) - \varepsilon \\
\geq & \sum\limits_{k\in \cK}  \frac{p_k (1-\varepsilon)}{ 1+\varepsilon } \int_{\Omega}   \theta \, d| D \chi_{D_k}| -\frac{2\varepsilon }{1+\varepsilon}.
\end{align*}
Since $\varepsilon$ is arbitrary, we have
$$
\int_\Omega \theta \, d  |D u | \geq \sum_{k\in\cK}p_k \int_\Omega \theta   \, d  | D \chi_{D_k} |.
$$
The inverse inequality is obvious. We have thus proved the assertion.
\end{proof}

\subsubsection{Total Energy Functional}\label{Section: total energy}

Based on Propositions~\ref{Prop: bulk enegy} and \ref{Prop: surface area}, 
the ensemble average  nonpolar energy can be formulated as
\begin{equation}
\label{nonpolar energy}
 \Inp(u) =    \int_\Omega \theta d|Du| + \int_\Omega \left[ P_h u   + \rho_s (1-u ) U^{\mathrm{vdW}}   \right] \, dx .  
\end{equation}
Therefore, the EASE functional is given by
\begin{align}\label{ensemble average total energy type II}
\hspace*{-2em}
\cE(u )  
=    \int_\Omega \theta d|Du| + \int_\Omega \left[ P_h u   + \rho_s (1-u ) U^{\mathrm{vdW}} \right] \, dx  
  + 
\int_\Omega   \left[    \rho \psi    - \frac{\epsilon(u) }{8\pi}|\nabla \psi|^2 - \chi_{\{u<1\}} B(\psi)         \right]\, dx
\end{align}
with admissible set $\cX$ and $\psi=\psi_u \in \cA$ is determined via the generalized PB equation~\eqref{GPBE2}.

\subsection{Size-modified Poisson-Boltzmann Theory}\label{Section: SMPB theory}

In the classic PB theory,  ionic solvents are assumed to be ideal: all ions and solvent molecules are point-like and there is no correlation between the (number) concentrations of these particles.
However, the ideal ionic solvent assumption overlooks the crucial  finite size effect of mobile ions and solvent molecules. 
Numerical simulations have demonstrated that the point-like ion assumption in classic PB theory can lead to nonphysically large ion concentrations near charged surfaces \cite{PhysRevLett.79.435}.

We will follow the idea in the pioneering work \cite{doi:10.1080/14786444208520813} by Bikerman  to derive the size-modified   polar energy formulation.
Assume that each  mobile ion  and   solvent molecule  occupies a finite volume $\sv$.
Building upon the definition of $\cC_0$ introduced in Section~\ref{Section: background}, we define  
$$
c_0^k= \bE(\cC_0| X=k) \quad \text{and}\quad c_0= \bE(\cC_0)= \sum_{k\in \cK} p_k c_0^k.
$$
It is assumed that inside the solvent phase, apart from the occupancy of ions, the remaining space is filled with solvent molecules. In this context, the following relation holds: 
\begin{equation}
\label{total concentration}
 \sv \sum_{j=0}^{N_c} c_j^k(x)= \chi_{U_k}(x) , \quad  x\in \Omega .
\end{equation}
This relation enforces a maximum concentration of $\sv^{-1}$ for each ion species. 
\eqref{total concentration} implies that
\begin{equation}
\label{solvent concentration}
c_0=  \sum_{k\in \cK} p_k c_0^k =  \sum_{k\in \cK} p_k  \left( \sv^{-1}- \sum_{j=1}^{N_c} c_j^k \right) \chi_{U_k}   
=  \sv^{-1} (1-u) - \sum_{j=1}^{N_c} c_j,
\end{equation}
Following \cite{doi:10.1080/14786444208520813}, we can include the solvent entropic contribution and then the size-modified entropy $S_{md}$ can be written as
\begin{equation}
\label{entropy-steric}
\begin{split}
-TS_{md}=&  \beta^{-1} \int_{\{ u<1 \}} \sum_{j=0}^{N_c} c_j \left[  \ln (\sv c_j ) -1 \right] \, dx \\
=& \beta^{-1} \int_{\{ u<1 \}} \sum_{j=1}^{N_c} c_j \left[  \ln (\sv c_j ) -1\right] \, dx  \\
& + \beta^{-1} \int_{\{ u<1 \}} \left(\sv^{-1} (1-u) - \sum_{j=1}^{N_c} c_j\right) \left[  \ln \left((1-u) -  \sv\sum_{j=1}^{N_c} c_j \right) -1\right] \, dx .
\end{split}
\end{equation}
See also \cite{PhysRevLett.79.435, BORUKHOV2000221, PhysRevE.75.021502, PhysRevE.75.021503} for some related work. 

Following the discussion in Remark~\ref{Rmk: PBE}, the variation of $H_{md}=\langle E_\psi \rangle - TS_{md}$, cf. \eqref{Def: internal energy ensemble}, with respect to $c_j$, $j=1,\cdots, N_c$,  gives
$$
\mu_j^\infty=   q_j \psi  + \beta^{-1} \left[ \ln (\sv c_j) -  \ln \left((1-u) -  \sv\sum_{i=1}^{N_c} c_i \right) \right]\chi_{\{u<1 \}}  .
$$
This leads to the relation
\begin{align*}
\frac{  c_j}{(1-u) -  \sv\sum_{i=1}^{N_c} c_i} = \chi_{\{u<1\}}  \sv^{-1} e^{\beta\mu_j^\infty}  e^{-\beta q_j \psi}.
\end{align*}
When $x$ is sufficient far away from the solute, we have $u(x)=0$, $\psi(x)=0$ and $c_j (x)=c_j^\infty$ with $j\in \{0,1,\cdots, N_c\}$,   where $c_0^\infty $ is the (constant) bulk concentration of the solvent molecule.
This implies that $c_j^\infty= c_0^\infty  e^{\beta\mu_j^\infty} $, where we have used the relation~\eqref{solvent concentration}.
\begin{equation}
\label{steric 1}
\frac{  c_j}{(1-u) -  \sv\sum_{i=1}^{N_c} c_i} = \chi_{\{u<1\}}   h_j e^{-\beta q_j \psi}, \quad \text{where } h_j =\frac{c_j^\infty}{\sv c_0^\infty }.
\end{equation}
Summing over $j=1,\cdots, N_c$, one can find out that
\begin{align*}
 \sum_{j=1}^{N_c} c_j =& (1-u) \frac{\chi_{\{u<1\}}  \sum_{j=1}^{N_c}  h_j e^{-\beta q_j \psi}}{1+ \chi_{\{u<1\}}  \sv \sum_{j=1}^{N_c}  h_j e^{-\beta q_j \psi}} \\
 =& (1-u) \frac{  \sum_{j=1}^{N_c}  c_j^\infty e^{-\beta q_j \psi}}{\sv c_0^\infty  +    \sv \sum_{j=1}^{N_c}   c_j^\infty e^{-\beta q_j \psi}} .
\end{align*}
Here, we have utilized the relation $(1-u)\chi_{\{u<1\}}=1-u$.
Note that the above equality implies that
\begin{equation}
\label{Kp}
1-u- \sv \sum_{j=1}^{N_c} c_j  = \frac{(1-u)\sv c_0^\infty }{\sv c_0^\infty +    \sv \sum_{j=1}^{N_c}   c_j^\infty e^{-\beta q_j \psi}}.
\end{equation}
Plugging this expression into \eqref{steric 1} yields
\begin{equation}
\label{Boltzmann distribution-steric}
c_j = (1-u)  \frac{ c_j^\infty e^{-\beta q_j \psi}}{\sv c_0^\infty +    \sv \sum_{i=1}^{N_c}   c_i^\infty e^{-\beta q_i \psi}}.
\end{equation}
Using \eqref{Boltzmann distribution-steric}, the size-modified Helmholtz free energy $H_{md}$ can be reformulated as
\begin{align*}
 H_{md}
=& \int_\Omega \left[ \rho \psi  - \frac{\epsilon(u)}{8\pi} |\nabla \psi|^2 + \beta^{-1} \sv^{-1} (1-u) \left( \ln \left( \frac{(1-u)\sv c_0^\infty }{\sv c_0^\infty +    \sv \sum_{j=1}^{N_c}   c_j^\infty e^{-\beta q_j \psi}} \right) -1  \right) \right]\, dx.
\end{align*}
A constant term (with respect to fixed $u$) 
\begin{equation}
\label{reference state ensemble}
 \beta^{-1} \sv^{-1} (1-u) \ln \left(  \sv c_0^\infty 	+    \sv \sum_{j=1}^{N_c}   c_j^\infty   \right) -  \beta^{-1} \sv^{-1} (1-u)  [\ln ((1-u)\sv c_0^\infty) -1]
\end{equation}
needs to be added to $H_{md}$ to adjust the reference state of the zero energy in the grand canonical ensemble under consideration.
We finally arrive at the following 
size-modified polar energy functional:
\begin{align}\label{ensemble average polar formulation-steric}
\begin{split}
\Ipd(u,\psi) =   \int_\Omega   \left[    \rho \psi    - \frac{\epsilon(u) }{8\pi}|\nabla \psi|^2 - (1-u)  B_{md} (\psi)       \right]\, dx,  
\end{split}
\end{align}
where the function $B_{md}:\R\to \R$ is defined by
\begin{equation}
\label{B steric}
B_{md} (s)=\beta^{-1} \sv^{-1} \ln \left( \frac{\sv c_0^\infty + \sv \sum_{j=1}^{N_c} c_j^\infty e^{-\beta q_j s}}{\sv c_0^\infty+ \sv \sum_{j=1}^{N_c} c_j^\infty } \right).
\end{equation}
Direct computations show
\begin{align*}
B_{md}' (s)= - \frac{\sum_{j=1}^{N_c} c_j^\infty q_j  e^{-\beta q_j s}}{\sv c_0^\infty + \sv \sum_{j=1}^{N_c} c_j^\infty   e^{-\beta q_j s}}
\end{align*}
and
\begin{align*}
B_{md}''(s)= \frac{\sv \beta c_0^\infty \sum_{j=1}^{N_c} c_j^\infty q_j^2 e^{-\beta q_j s} + \sv \beta \left( \sum_{j=1}^{N_c} c_j^\infty q_j^2 e^{-\beta q_j s} \right) \left( \sum_{j=1}^{N_c} c_j^\infty   e^{-\beta q_j s} \right) -\sv \beta \left( \sum_{j=1}^{N_c} c_j^\infty q_j  e^{-\beta q_j s} \right)^2 }{ \left( \sv c_0^\infty + \sv \sum_{j=1}^{N_c} c_j^\infty   e^{-\beta q_j s} \right)^2}.
\end{align*}
The Cauchy-Schwartz inequality shows that $B_{md}''(s)> 0$ for all $s\in \R$. Therefore, $B_{md}$ is strictly convex.
It follows from \eqref{neutral} that $B_{md}'(0)=0 $ and thus $0=B_{md}(0)=\min_{s\in \R} B_{md}(s)$.
Plugging \eqref{Boltzmann distribution-steric} into \eqref{Poisson eq grand} gives the size-modified PB equation
\begin{equation}
\label{SMPBE}
\begin{split}
\nabla \cdot [  \epsilon(u) \nabla \psi ] -4\pi(1-u)B_{md}'(\psi) + 4\pi \rho =0 \quad \text{in }\Omega. 
\end{split}
\end{equation}
Lemma~\ref{Lem: GPBE estimate} shows that, for any fixed $u\in \cX$, $\psi\in \cA$ solves \eqref{SMPBE} iff it maximizes $\Ipd(u,\cdot)$.

Define $\sL_{md}: \cX\to \R$ by 
$$
 \sL_{md}(u )=  \int_\Omega \left[ P_h u + \rho_s (1-u)  U^{\mathrm{vdW}}) \right]\, dx  + \Ipd(u,\psi_u) ,
$$
where $\psi_u$   maximizes $\Ipd(u,\cdot)$ in $\cA$.
Following the proof of Lemma~\ref{Lem: modify sets}, we can prove a similar result for $\sL_{md}$.
\begin{lem}\label{Lem: modify sets-size}
Let $\{D_k\}_{k\in \cK}$
be a family of    Caccioppoli sets satisfying $\Oi  \subset D_k \subset \Omega\setminus \overline{\Omega}_{\rm e}    $  and $p_k\in [0,1]$ with $\sum_{k\in \cK} p_k=1$. 
Then for each $\varepsilon>0$,
there exists another  family $\{\widetilde{D}_k \}_{k\in \cK}$ of Caccioppoli sets  satisfying $\Oi  \subset \widetilde{D}_k \subset \Omega\setminus \overline{\Omega}_{\rm e}    $ and
\begin{align*}
\cH^2(\partial^* \widetilde{D}_k \cap \partial^* \widetilde{D}_j )=0 ,\quad \forall k,j\in \cK,\, \, k\neq j.
\end{align*}
Moreover,        
\begin{equation}\label{eq: modify energy}
\left| \sum\limits_{k\in \cK}p_k \int_\Omega \theta d |D \chi_{D_k}|  +   \sL_{md} (u  )  - \sum\limits_{k\in \cK}p_k \int_\Omega \theta d |D \chi_{\widetilde{D}_k}|   - \sL_{md} (\widetilde{u} )     \right|<\varepsilon , 
\end{equation}
where 
$u=\sum\limits_{k\in \cK}p_k \chi_{D_k}$ and $\widetilde{u}=\sum\limits_{k\in \cK}p_k \chi_{\widetilde{D}_k}.$
\end{lem}

By virtue of Lemma~\ref{Lem: modify sets-size} and Proposition~\ref{Prop: surface area}, the   size-modified EASE functional can be expressed as
\begin{equation}
\label{ensemble average solvation energy-steric}
\begin{split}
  \cE_{md}(u)   =  \int_{\Omega} \theta d|Du| +    \int_\Omega \!  \left[   P_h u    +  (1-u ) \left(\rho_s U^{\mathrm{vdW}}  - B_{md}(\psi) \right) +  \rho \psi    - \frac{\epsilon(u) }{8\pi} |\nabla \psi|^2             \right]  dx 
\end{split}
\end{equation}
with admissible set $\cX$ and $\psi=\psi_u \in \cA$ solves the size-modified PB equation~\eqref{SMPBE}.

\section{Analysis of Ensemble Average  Solvation Energy}\label{Secion: analysis of ensemble average energy}

We will perform some preliminary analysis for the EASE models developed in the previous section. 
The following example shows that $\cE$ is not l.s.c. with respect to   convergence in $W^{1,1}(\Omega)$.

\begin{example}
We take $\psi_D \equiv 1$, $\Oi=B(0,1)$, $\Ot= \{x\in \R^3: 1<|x|<2\}$ and  $\Oe= \{x\in \R^3: 2<|x|<3\}$. Choose $u\in \cX \cap C^\infty(\Omega)$ such that $u(x)=1$ for $|x|\leq \frac{5}{3}$ and $u(x)<1$ elsewhere. We can construct a sequence $u_n \in \cX\cap W^{1,1}(\Omega)$ such that $u_n\to u$ in $W^{1,1}(\Omega)$ satisfying
\begin{itemize}
\item $u_n(x) <1$ for all $|x|>1$,
\item $u_n(x)=\frac{n-1}{n}$ for $\frac{4}{3}\leq |x| \leq \frac{5}{3}$.
\end{itemize}
For instance, we can define $u_n(x)= \frac{n+3}{n}- \frac{3}{n}|x|$ for $1\leq |x|\leq \frac{4}{3}$ and $u_n(x)= \frac{n-1}{n}u(x)$ for $|x|>\frac{4}{3}$.
Direct computations show that 
$$
\lim\limits_{n\to \infty} \Inp(u_n)= \Inp(u).
$$
However, Lemma~\ref{Lem: GPBE converge} shows that
\begin{equation}
\label{ellipticcontinuation}
\lim\limits_{n\to \infty} \Ip(u_n,\psi_{u_n}) =    
\int_\Omega   \left[    \rho \psi_*    - \frac{\epsilon(u) }{8\pi}|\nabla \psi_*|^2 - \chi_{\{|x|>1\}} B(\psi_*)         \right]\, dx \leq  \Ip(u, \psi_*) \leq \Ip(u,\psi_u),
\end{equation}
where  $\psi_u\in \cA$  is the solution of \eqref{GPBE2} and $\psi_{u_n}$ is the solution with $u$ replaced by $u_n$.
Additionally, $\psi_*$ is the solution of 
\begin{equation}
 \nabla \cdot [  \epsilon(u) \nabla \psi ] -4\pi \chi_{\{|x|>1\}}B'(\psi) + 4\pi \rho=0 \quad \text{in }\Omega.	
\end{equation}
Note that equalities hold in \eqref{ellipticcontinuation} iff $\psi_u=\psi_*$ and $\psi_*(x)=0$ for a.a. $1<|x|\leq \frac{5}{3}$. 
If the second condition holds, then $\psi_*$ solves
\begin{equation}
\label{PBE-UC}
\left\{\begin{aligned}
 \nabla \cdot [  \epsilon(u) \nabla \psi ] -4\pi B'(\psi) &=0  &&\text{in}&& \{ 1<|x|<3\} ,\\
\psi&=0 &&\text{on}&& \{ |x|=1\} ,\\
\psi & = 1 &&\text{on}&& \partial\Omega. \\
\end{aligned}\right.
\end{equation} 
Then it follows from the elliptic unique continuation theorem, c.f. \cite[Theorem~1.4]{Choullielliptic}, that $\psi_*\equiv 0$, which contradicts with the fact that $\psi_*=1$ on $\partial\Omega$.
\end{example}

The analysis above reveals that the absence of lower semicontinuity in $\cE$ arises from the discontinuous term $\chi_{\{u<1 \}}$. Minimizing non-lower semicontinuous functionals poses a significant challenge.
Particularly, the  computational and analytical complexities associated with the discontinuous term further compound the difficulty. Consequently, the minimization of $\cE$ remains an open problem.

In contrast, the  discontinuous term $\chi_{\{u<1\}}$ is obviated in $\cE_{md}$ through the relation $(1-u)\chi_{\{u<1\}}=1-u$.
This elimination allows us to establish the lower semicontinuity of $\cE_{md}$, as demonstrated in the proof of Theorem~\ref{Thm: wellposedness} below. 
This paves the way for establishing additional analytic properties of $\cE_{md}$.
Consequently, our focus in the remainder of this article will center on the analysis of the size-modified EASE functional $\cE_{md}$.
The analyses in the following two subsections underscore  the advantages of incorporating finite size effects in EASE modeling from an analytical perspective.

\subsection{Existence of Minimizer}

\begin{theorem}
\label{Thm: wellposedness}
$\cE_{md} $ has a global minimizer in $\cX$.
\end{theorem}
\begin{proof}
First, we will show that $\cE_{md}$ is convex.  
Indeed, given any $u_0,u_1 \in \cX $, let $u_t= t u_0 + (1-t)u_1$ for $t\in [0,1]$. We have $u_t\in \cX$ for all $t\in [0,1]$. 
For any  $v\in \cX$, let $\psi_v$ be the solution  of \eqref{SMPBE} in $\cA$ with $u$ replaced by $v$.
Note that, for any fixed $\psi\in \cA$, both $\Inp(\cdot)$ and $\Ipd(\cdot, \psi)$ are convex.
Then it follows from Lemma~\ref{Lem: GPBE estimate}  that
\begin{align*}
\cE_{md}(u_t) & =  \Inp(u_t)+ \Ipd(u_t, \psi_{u_t})   \\
&\leq    t \left[ \Inp(u_0)+\Ipd(u_0, \psi_{u_t}) \right]  + (1-t) \left[ \Inp(u_1)+\Ipd(u_1, \psi_{u_t}) \right]   \\
&\leq  t \left[ \Inp(u_0)+\Ipd(u_0, \psi_{u_0}) \right]  + (1-t) \left[ \Inp(u_1)+ \Ipd(u_1, \psi_{u_1}) \right] \\
& =t \cE(u_0) + (1-t) \cE(u_1). 
\end{align*}
The lower semicontinuity of $\cE_{md}$ in $\cX$ with respect to convergence in $L^1(\Omega)$  is a direct conclusion from the dominated convergence theorem,
\cite[Corollary~1]{MR2345913}, and Lemma~\ref{Lem: GPBE converge}.
Now the existence of a minimizer of $\cE_{md}$ in $\cX$   can be immediately obtained by the direct method of Calculus of  Variation.
\end{proof}

\begin{remark}\label{Rmk: Lip wellposedness}
The assertion of Theorem~\ref{Thm: wellposedness} remains valid if we merely assume that $\Sigma_0$ and $\Sigma_1$ are Lipschitz continuous.
\end{remark}

\subsection{Continuous dependence on the constrains}

Let $\srE=\min_{u\in \cX}\cE_{md}(u)$. 
Since $\srE$ depends on  Constraint~\eqref{constrain 2} and thus on $\Sigma_0$ and $\Sigma_1$, we will   justify the robustness of \eqref{ensemble average solvation energy-steric}  by demonstrating that $\srE$  depends continuously   on   $\Sigma_j$ with $j=0,1$  in a proper topology. Without loss of generality, we may assume that $\Sigma_j$ are connected. The general case can be proved in a similar manner.


For $k=1,2$, let $\mathcal{MH}^k$ be the set of  pairs $(\Sigma_0,\Sigma_1)$ of compact, connected and oriented $C^k$-hypersurfaces without boundary contained in the open set $\Omega$, which induces a decomposition $(\Oi, \Ot,\Oe)$ as in Figure~\ref{fig:sharp}(B). 
The orientation  of $\Sigma_0$ ($\Sigma_1$, resp.) is so taken that its outward unit normal vector field  $\nu_{\Sigma_0}$ ($\nu_{\Sigma_1}$, resp.)    points into $\Ot$.
To summarize, the open sets   $(\Oi, \Ot,\Oe)$ satisfy the following properties.
\begin{itemize}
\item $(\Oi, \Ot,\Oe)$ are connected so that $\Sigma_0 \cap \Sigma_1 =\emptyset$.
\item   $ \bigcup_{i=1}^{N_m} \overline{B}(x_i,\sigma_i) \subset  \Oi \Subset   \Omega\setminus \overline{\Omega}_{\rm e}$. 
\item $\partial\Omega \subset \partial\Oe$.
\end{itemize}
Theorem~\ref{Thm: wellposedness} and Remark~\ref{Rmk: Lip wellposedness} show that  $\srE$ can be considered  a functional defined on $\mathcal{MH}^1$. 
In the rest of this section, we  denote $\cE_{md}$, $\cX$, and $(\Oi, \Ot,\Oe)$ by  $\cE_{md,\Gamma}$, $\cX_{\Gamma}$, and $(\Oig, \Otg,\Oeg)$ to indicate their dependence on $\Gamma=(\Sigma_0,\Sigma_1)\in \mathcal{MH}^1$. 
In Theorem~\ref{Thm: continuous dep on constraints} below, it will be shown that $\srE$ is indeed continuous while restricted to  $\mathcal{MH}^2$.  
To this end, we will first demonstrate that $\mathcal{MH}^k$ are metric spaces.

The normal bundle of a compact, connected and oriented $C^1$-hypersurface $\Sigma$ without boundary is given by
$$
\cN^1   \Sigma=\{(\q, \nu_{\Sigma}(\q) ) : \, \q\in \Sigma \}\subset \R^3\times \R^3   ,
$$
where $\nu_{\Sigma}(\q)$ is the outward unit normal of $\Sigma$ at $\q\in\Sigma$.
If in addition, $\Sigma$ is $C^2$, its second normal bundle can be defined as      
$$
\cN^2   \Sigma=\{(\q, \nu_{\Sigma}(\q), \nabla_{\Sigma}\nu_{\Sigma}(\q) ) : \, \q\in \Sigma \}\subset \R^3\times \R^3  \times \R^6,
$$
where $\nabla_{\Sigma}$ is the surface gradient of $\Sigma$.
Recall that the Hausdorff metric on compact subsets $K\subset\R^n$  is defined by
$$
d_{\cH}(K_1,K_2)=\max\{ \sup\limits_{x\in K_1}d(x,K_2), \sup\limits_{x\in K_2}d(x,K_1) \}.
$$
We can equip $\mathcal{MH}^k$ with the metric
$$
d_{\mathcal{MH}^k}( (\Sigma_0,  \Sigma_1), (\widetilde{\Sigma}_0,  \widetilde{\Sigma}_1) )=d_{\cH}(\cN^k   \Sigma_0, \cN^k   \widetilde{\Sigma}_0) +d_{\cH}(\cN^k  \Sigma_1 , \cN^k   \widetilde{\Sigma}_1)  , \quad (\Sigma_0,  \Sigma_1), (\widetilde{\Sigma}_0,  \widetilde{\Sigma}_1) \in \mathcal{MH}^k .
$$
Following \cite[Chapter 2]{MR3524106}, one can show that $\mathcal{MH}^2 $ equipped with $d_{\mathcal{MH}^2}$ is a Banach manifold.
The   theorem below is the main result of this section.
\begin{theorem}\label{Thm: continuous dep on constraints}
$\srE\in C (\mathcal{MH}^2)$.
\end{theorem}

We will split the proof of Theorem~\ref{Thm: continuous dep on constraints} into several steps.
Pick arbitrary $\Gamma=(\Sigma_0,\Sigma_1)\in \mathcal{MH}^2$.
To prove the continuity of $\srE$ at $\Gamma$, we choose an arbitrary sequence $\Gamma_n:=(\Sigma_{0,n},\Sigma_{1,n}) \in \mathcal{MH}^1$. 
In the sequel, we will first prove two propositions concerning the convergence $ \lim\limits_{n\to \infty} \srE(\Gamma_n) = \srE(\Gamma) $ under the condition
\begin{equation}
\label{sequence hypersurface}
\lim\limits_{n\to \infty} d_{\mathcal{MH}^1} (\Gamma_n,  \Gamma )=0  
\end{equation}
when $\Gamma_n$ converges to $\Gamma$ from the interior and exterior of $\Gamma$ (with respect to the orientations of $(\Sigma_0,\Sigma_1)$).
In the rest of this section, we will denote  by $\um$ ($\umn$, resp.) a global minimizer of $\cE_{md, \Gamma}$ ($\cE_{md, \Gamma_n}$ resp.) in $\cX_\Gamma $ ($\cX_{\Gamma_n} $ resp.).

\begin{prop}\label{Prop: inner approx}
Assume that $\Gamma=(\Sigma_0,\Sigma_1)\in \mathcal{MH}^2$ and $\Gamma_n=(\Sigma_{0,n},\Sigma_{1,n})\in \mathcal{MH}^1$ satisfy \eqref{sequence hypersurface}.
If
\begin{equation}
\label{sequence domain}
\Oign \subset \Oig \quad \text{and} \quad \Oegn \subset \Oeg ,
\end{equation}
then $ \lim\limits_{n\to \infty} \srE(\Gamma_n) = \srE(\Gamma) $. 
\end{prop}
\begin{proof}
Observe that $\um \in \cX_{\Gamma_n}$ for all $n$. Thus
$$
\srE(\Gamma_n) = \cE_{md, \Gamma_n} (\umn) \leq \cE_{md, \Gamma_n}(\um)=\cE_{md, \Gamma}(\um) =\srE(\Gamma).
$$
It follows from Lemma~\ref{Lem: GPBE estimate} that there exists some constant $M>0$ such that
$$
 \left| \int_{\Omega }\rho_s  U^{\rm vdW}  (1-\umn) \, dx + \Ipd(\umn, \psi_{\umn}) \right| \leq M,
$$
where $\psi_{\umn}\in \cA$ is the solution of \eqref{SMPBE} with $u$ replaced by $\umn$.
This implies that
\begin{align*}
\theta_0\int_\Omega  d|D \umn| +    P_h \|\umn\|_{L^1}  -M  \leq \srE(\Gamma).
\end{align*}
Therefore, $\| \umn\|_{BV}$ is uniformly bounded.
\cite[Theorem~5.2.3.4]{MR1158660} implies that 
there exists a  subsequence, not relabelled, and some $u\in BV(\Omega)$  such that $\umn\to u$ in $L^1(\Omega)$. Note that \eqref{sequence hypersurface}  implies that
$$
\lim\limits_{n\to \infty}\left\|  \chi_{\Oign}-\chi_{\Oig } \right\|_{L^1} =\lim\limits_{n\to \infty}\left\|  \chi_{\Oegn}-\chi_{\Oeg } \right\|_{L^1}=0
$$
and thus $u\in \cX_\Gamma$.
From \cite[Corollary 1]{MR2345913}, Lemma~\ref{Lem: GPBE converge}  and the dominated convergence theorem, we infer that
\begin{align*}
\srE(\Gamma) \leq \cE_{md,\Gamma}(u)\leq \liminf\limits_{n\to \infty} \srE(\Gamma_n) \leq \limsup\limits_{n\to \infty} \srE(\Gamma_n) \leq \srE(\Gamma).
\end{align*}
This proves the  assertion.
\end{proof}

To prove the convergence $ \lim\limits_{n\to \infty} \srE(\Gamma_n) = \srE(\Gamma) $ in case
\begin{equation}
\label{sequence domain}
\Oign \supseteq \Oig \quad \text{and} \quad \Oegn \supseteq \Oeg  ,
\end{equation}
we will need some preparations. 
Because $\Sigma_i$ are $C^2$, following the proof of Lemma~\ref{Lem: modify sets}, the map
$$
\Lambda_i : \Sigma_i\times (-\a,\a) \to B_\a(\Sigma_i) : \,\, (\q, r) \mapsto \q+ r \nu_{\Sigma_i}(\q), \quad i=0,1,
$$
is a $C^1$-diffeomorphism for some $\a>0$, where $B_\a(\Sigma_i)$ is the tubular neighbourhood of $\Sigma_i$ with width $\a$ and  $\nu_{\Sigma_i}$ is 
the  outward unit normal of $\Sigma_i$. 
By the inverse function theorem, there exist  two  maps $P_i\in C^1(B_{\a}(\Sigma_i), \Sigma_i)$ and $d_i\in C^1(B_{\a}(\Sigma_i),(-\a,\a))$, where $P_i$ is the nearest point projection onto $\Sigma_i$ and $d_i$ is the signed distance to $\Sigma_i$ with $d_i(x)>0$ for $x\in B_{\a}(\Sigma_i) \cap \Ot$.
Note that $\cG_i^r:=\Lambda_i (\Sigma_i,r) \in \mathcal{MH}^1$. Its orientation is chosen to be compatible with that of $\Sigma_i$ so that its
outward unit normal 
$$
\nu_{i,r}(x)=\nu_{\Sigma_i}(P_i(x)) ,\quad  x\in \cG_i^r.
$$
See \cite[Section 2.2.2]{MR3524106}.
Then it is clear that $\lim\limits_{r\to 0} d_{\cH}( \cG_i^r,  \Sigma_i)=0.$
Without loss of generality, we may assume that  $\a>0$ is so small that for any $r\in (-\a,\a)$, $ \cG_0^r \cap \cG_1^r =\emptyset$.

\begin{lem}\label{tech lem 2}
For every $f\in \cX_{\Gamma}$, we define  
\begin{align*}
f_r(x)=
\begin{cases}
1 ,\quad & x\in \Oi \cup B_r(\Sigma_1)\\
0, & x\in \Oe \cup B_r(\Sigma_0)\\
f(x), & \text{elsewhere} 
\end{cases}
\end{align*}
for $r\in (0,\a)$.
Then $f_r\to f$ in $L^1(\Omega)$ and $\int_\Omega \theta d |D f_r(x)| \to \int_\Omega \theta d |D f (x)|$ as $r\to 0^+$.
\end{lem}
\begin{proof}
The proof for $f_r\to f$ in $L^1(\Omega)$ is straightforward. So we will only show the second part. 
\cite[Corollary~1]{MR2345913} implies that
$ 
\int_\Omega \theta d|Df(x)| \leq \liminf\limits_{r\to 0^+} \int_\Omega \theta d|Df_r(x)| .
$ 
In the rest of the proof, it is   assumed that $i\in \{0,1\}$.
Put $U_i^r:=B_r(\Sigma_i) \cap \Ot$.
For any $\phi\in C^1_c(\Omega;\R^3)$ with $\|\phi\|_\infty \leq 1$ and $u\in \cX_{\Gamma}$,   the trace theorem of $BV$-functions, cf. \cite[Theorem~5.3.1]{MR1158660} implies  that
\begin{equation*} 
\int_{U_i^r} u \nabla \cdot (\theta \phi) \, dx + \int_{U_i^r}  (\theta \phi) \cdot d[Du]   =\int_{\cG_i^r } (\theta \phi)\cdot \nu_{i,r} T_r u \, d\cH^2- \int_{\Sigma_i } (\theta \phi)\cdot \nu_{\Sigma_i} T_r u\, d\cH^2. 
\end{equation*}
Here $T_r u$ is the trace of $u|_{U_i^r}$  on $\partial U_i^r$; and $[Du] $ is the vector-valued measure for the gradient of $u$. 
Pushing $r\to 0^+$ above yields that
\begin{equation}\label{trace eq 2}
\lim\limits_{r\to 0^+}  \int_{\cG_i^r } (\theta \phi)\cdot \nu_{i,r} T_r u \, d\cH^2 = \int_{\Sigma_i } (\theta \phi)\cdot \nu_{\Sigma_i} T_r u\, d\cH^2.
\end{equation}
Direct computations show	
\begin{align*}
\int_{\Omega} (\theta \phi) \cdot d[D f_r] & =    \sum_{i=0,1} \int_{\cG_i^r} (\theta \phi) \cdot d[D f_r]  +    \int_{ \Ot\setminus (U_0^r \cup U_1^r) } (\theta \phi) \cdot d[D f]  \\
&= \int_{\Omega} (\theta \phi) \cdot d[D f ]  - \sum_{i=0,1} \int_{U_i^r} (\theta \phi) \cdot d[D f] - \sum_{i=0,1} \int_{\Sigma_i} (\theta \phi) \cdot d[D f] \\
& \quad - \sum_{i=0,1} \int_{\cG_i^r} (\theta \phi) \cdot d[D (f- f_r)]\\
&\leq \int_{\Omega} \theta   d| D f (x)|  - \sum_{i=0,1} \int_{U_i^r} (\theta \phi) \cdot d[D f] - \sum_{i=0,1} \int_{\Sigma_i} (\theta \phi) \cdot d[D f] \\
& \quad - \sum_{i=0,1} \int_{\cG_i^r} (\theta \phi) \cdot d[D (f- f_r)].
\end{align*}
It follows from \eqref{trace eq 2} that
\begin{align*}
&  -\sum_{i=0,1} \int_{\Sigma_i} (\theta \phi) \cdot d[D f]  - \sum_{i=0,1} \int_{\cG_i^r} (\theta \phi) \cdot d[D (f- f_r)] \\
=&  \sum\limits_{i=0,1} \int_{\Sigma_i} (\theta \phi) \cdot \nu_{\Sigma_i} (i- T_r f) \, d\cH^2 - \sum\limits_{i=0,1} \int_{\cG_i^r} (\theta \phi)\cdot \nu_{i,r} (i- T_r f) \, d\cH^2  
\to   0  
\end{align*}
as  $r\to 0^+$. 
See the proof of \cite[Theorem~5.4.1]{MR1158660}.
We clearly have
$ 
\lim\limits_{r\to 0^+}  \int_{U_i^r} (\theta \phi) \cdot d[D f]   =0.
$ 
Therefore, we can infer from \eqref{Def: weighted total variation} and \cite[Theorem~5.3.1]{MR1158660} that
$$ 
\limsup\limits_{r\to 0^+} \int_\Omega \theta d|Df_r(x)| \leq \int_\Omega \theta d|Df(x)| .
$$
This proves the assertion
\end{proof}

\begin{prop}\label{Prop: outer approx}
Assume that $\Gamma=(\Sigma_0,\Sigma_1)\in \mathcal{MH}^2$ and $\Gamma_n=(\Sigma_{0,n},\Sigma_{1,n})\in \mathcal{MH}^1$ satisfy \eqref{sequence hypersurface}.
If
\begin{equation}
\label{sequence domain-2}
\Oign \supseteq \Oig \quad \text{and} \quad \Oegn \supseteq \Oeg  ,
\end{equation}
then $ \lim\limits_{n\to \infty} \srE(\Gamma_n) = \srE(\Gamma) $.  
\end{prop}
\begin{proof}
Let $r_n= d_{\mathcal{MH}^1} (\Gamma_n,\Gamma)$. 
For sufficiently large $n\in \N$, we define
\begin{align*}
u_n (x)=
\begin{cases}
1 ,\quad & x\in \Oi \cup B_{r_n}(\Sigma_1),\\
0, & x\in \Oe \cup B_{r_n}(\Sigma_0) ,\\
\um (x) , & \text{elsewhere}.
\end{cases}
\end{align*}
Since $u_n \in \cX_{\Gamma_n}$ and $\umn \in \cX_\Gamma$, we have
\begin{equation}
\label{domain dep converg 1}
\srE(\Gamma) \leq \cE_{md,\Gamma} (\umn) =\cE_{md,\Gamma_n} (\umn) =\srE(\Gamma_n) \leq \cE_{md,\Gamma_n} (u_n)  = \cE_{md,\Gamma} (u_n) .
\end{equation}
Lemma~\ref{tech lem 2}  implies that as $n\to \infty$ 
$$
u_n \to \um  \quad \text{in }L^1(\Omega) \quad \text{and}\quad \int_\Omega \theta d|Du_n | \to \int_\Omega \theta d|D\um |.
$$
Using the dominated convergence theorem and Lemma~\ref{Lem: GPBE converge} yields
$ 
\lim\limits_{n\to \infty} \cE_{md,\Gamma}(u_n )=\cE_{md,\Gamma}(\um )=\srE(\Gamma).
$ 
Then the asserted statement follows by pushing $n\to \infty$ in \eqref{domain dep converg 1}.
\end{proof}

\begin{proof}[Proof of Theorem~\ref{Thm: continuous dep on constraints}]
Now we assume that  $\Gamma_n\in \mathcal{MH}^2$ and
$$
\Gamma_n:=(\Sigma_{0,n},\Sigma_{1,n}) \to \Gamma \quad \text{in } \mathcal{MH}^2 .
$$
Let $r_n= d_{\mathcal{MH}^2} (\Gamma_n,\Gamma)$.
Then  $\cG^{\pm r_n} = (\cG_0^{\pm r_n}, \cG_1^{\pm r_n})\in \mathcal{MH}^1$.
We denote by $u_n^{\pm}$ a global minimizer of $\cE_{md, \cG^{\pm r_n}}$ in $\cX_{\cG^{\pm r_n}}$.
Because $\umn\in \cX_{\cG^{- r_n}}$ and $u_n^+ \in \cX_{\Gamma_n}$, we infer that
$$
\srE(\cG^{- r_n} )\leq \cE_{md, \cG^{- r_n}} (\umn)= \cE_{md,\Gamma_n } (\umn)=\srE(\Gamma_n) \leq \cE_{md,\Gamma_n } (u_n^+) = \cE_{md, \cG^{+ r_n}} (u_n^+) = \srE(\cG^{+ r_n} ).
$$
Propositions~\ref{Prop: inner approx} and \ref{Prop: outer approx} show that
$$
\lim\limits_{n\to \infty} \srE(\cG^{- r_n} ) = \lim\limits_{n\to \infty} \srE(\cG^{+ r_n} ) = \srE(\Gamma).
$$
Therefore, $\lim\limits_{n\to \infty} \srE(\Gamma_n)= \srE(\Gamma).$
\end{proof}



\section{Conclusion}\label{Section: Conclusion}


Variational implicit solvation models (VISM) have been successful in the computation of solvation energies. However, traditional sharp-interface VISMs do not account for the inherent randomness of the solute-solvent interface, stemming from thermodynamic fluctuations. It has been demonstrated that neglecting these fluctuations can lead to substantial errors in predicting solvation free energies.

In this work, a new approach to calculating ensemble averaged solvation energy (EASE) is developed using diffuse-interface VISM. Grounded in principles of statistical mechanics and geometric measure theory, the method effectively captures EASE during the solvation process by employing a novel diffuse-interface profile $u(x)$, which represents the probability of finding a point $x$ in the solute phase across all microstates within the grand canonical ensemble.
To showcase the versatility of our modeling paradigm, we formulate two EASE functionals: one within the classic Poisson-Boltzmann (PB) framework and another within the framework of size-modified PB theory, accounting for finite-size effects of mobile ions and solvent molecules.
 
In the routine calculation of the EASE, one needs to carry out molecular dynamics (MD) simulations to obtain thousands of solute-solvent configures (snapshots) and perform energy calculations for each snapshot. By meticulously modeling the impact of conformational changes in the solvent medium, the proposed model can reproduce EASE with just one diffuse-interface configuration, potentially drastically speeding up calculations compared to ensemble-averaging energies from thousands of snapshots.

Preliminary analyses of the proposed EASE models indicate that the size-modified EASE functional outperforms the EASE functional in the classic PB theory in various analytical aspects. Motivated by these observations, we plan to conduct numerical implementations and further theoretical analyses of the size-modified EASE functional in future work.


\appendix

\section{ A class of strictly convex functionals}\label{Appendix A}

In this section, we collect some results concerning a class of strictly convex functionals associated with  the polar solvation energy .
These results can be proved by following the proofs of \cite[Propositions~3.1 and 3.2]{MR4523426} line by line.
We will thus omit their proofs.

Throughout this appendix, we assume that  $F\in C^\infty(\R)$ is a strictly convex function with $F'(0)=0$ and $F(0)=0$.

\begin{lem}\label{Lem: GPBE estimate}
Assume that $a,b,c\in L^\infty(\Omega)$ satisfy
\begin{equation}
\label{bdd of coeff}
0<L_0 \leq a \leq L_1  ,\quad 0\leq b \leq L_2 , \quad \|c\|_\infty \leq L_3
\end{equation}
for some constants $L_i$.
Then the functional 
$$
G(\psi)=\int_\Omega \left[ \frac{a}{2}|\nabla \psi|^2 +b F(\psi) -c \psi  \right]\, dx
$$
has a unique minimizer $\psi\in \cA$ for every $\psi_D\in W^{1,\infty}(\Omega)$, c.f. \eqref{Def cA}, or equivalently, $\psi$ weakly solves
\begin{equation*}
\left\{\begin{aligned}
\nabla\cdot(a  \nabla \psi)  -b F'(\psi) +c &=  0  &&\text{in}&&\Omega;\\
\psi&=\psi_D   &&\text{on}&&\partial\Omega.
\end{aligned}\right.
\end{equation*} 
Moreover,  for some constant $C_\infty$   depending only on  $L_i$  and $\psi_D$
$$
\|\psi \|_{H^1} + \|\psi \|_{L^\infty} \leq C_\infty .
$$
\end{lem}

\begin{lem}\label{Lem: GPBE converge}
Let $a_n,b_n,c \in L^\infty(\Omega)$ satisfy \eqref{bdd of coeff}, $n=0,1,\cdots$.
Assume that $\psi_n $ is the unique minimizer of
$$
G_n(\psi)=\int_\Omega \left[ \frac{a_n}{2}|\nabla \psi|^2 +b_n F(\psi) -c \psi  \right]\, dx
$$
in $\cA$ for some $\psi_D\in W^{1,\infty}(\Omega)$.
If $a_n\to a_0$ and $b_n\to b_0$ in 
$L^1(\Omega)$  as  $n\to \infty$,
then
$$
\psi_n \to \psi_0 \quad \text{in } H^1(\Omega),\quad \text{and}\quad \lim\limits_{n\to \infty} G_n(\psi_n)= G_0(\psi_0).
$$
\end{lem}


\section{A comparison with GFBVISM}\label{Appendix B}

In this section, we will conduct a comparative analysis of the proposed models~\eqref{ensemble average total energy type II} and \eqref{ensemble average solvation energy-steric} with a closely related VISM known as the Geometric Flow-Based VISM (GFBVISM), which is rooted in classic PB theory.   GFBVISM  \cite{MR2719171}   is defined as follows:
\begin{align}\label{ensemble average total energy type I}
\mathcal{E}^{(2)}(u )  
= \Inp(u)+\Ip^{(2)}(u,\psi),
\end{align}
where
\begin{align}\label{ensemble average polar formulation 1}
\Ip^{(2)}(u,\psi)=  \int_\Omega   \left[    \rho \psi    - \frac{\epsilon(u) }{8\pi}|\nabla \psi|^2 - (1-u )  B(\psi)         \right]\, dx, 
\end{align}
and $\psi\in \cA$ solves
\begin{equation}
\label{GPBE1}
 \nabla \cdot [  \epsilon(u) \nabla \psi ] -4\pi(1-u)B'(\psi) + 4\pi\rho=0 \quad \text{in }\Omega.	
\end{equation}
Despite certain similarities between \eqref{ensemble average total energy type II} and \eqref{ensemble average solvation energy-steric} and GFBVISM, there exist fundamental distinctions between them.

First, \eqref{ensemble average total energy type I} was introduced in an {\em ad hoc} way to create a transition region between the solute and  solvent, lacking an explanation of the physical meanings of the transition parameter $u$ and the   predicted energy.

Second, it has come to attention that the original formulation of GFBVISM \cite{MR2719171} does not incorporate Constraints \eqref{constrain 1} and \eqref{constrain 2}. This absence of constraints introduces the possibility of nonphysical minimizers within the GFBVISM model. For instance, it may allow for trivial values of $u$ or even negative values, which are not physically meaningful in the given context. Therefore, including these constraints becomes crucial to ensure the model's solutions align with physical principles.

Third and most importantly, the proposed models correct and improve  the derivation of the ensemble average polar energy.  
Due to Proposition~\ref{Prop: bulk enegy}, one may guess that $\Ip^{(2)}$ approximates the ensemble average polar energy. 
However,   Formulation~\eqref{ensemble average polar formulation 1} is questionable in the sense that it is derived from an erroneous ``entropy" formulation.
To see this, on   a heuristic level,  one can ensemble average the entropy and obtain 
$$
-T \langle S \rangle = -T\sum_{k\in \cK}p_k S_k= \beta^{-1} \sum_{k\in \cK}p_k \int_\Omega  \sum\limits_{j=1}^{N_c} c_j^k  \left[  \ln(\sv c_j^k  ) -1  \right]\, dx \quad \text{and}\quad \langle H \rangle : = \langle E \rangle  -T \langle S \rangle,
$$
where $S_k$ is defined in \eqref{Def: entropy}.
By taking the first variations of $\langle H \rangle  $   with respect to all $c_j^k$, a similar argument to Remark~\ref{Rmk: PBE} gives \eqref{ion concentration and Boltzmann distribution}.
It follows from \eqref{mean ion con} that
\begin{equation}
\label{cj formulation 1}
c_j(x)=   \sum\limits_{k\in \cK} p_k c_j^k(x)= \sum\limits_{k\in \cK} p_k \chi_{U_k}(x) c_j^k(x)= (1-u(x))  c_j^\infty e^{-\beta q_j \psi (x)} .
\end{equation}
See \eqref{Def: entropy} for the definition of $S_k$.
Plugging \eqref{cj formulation 1} into \eqref{Poisson eq} yields \eqref{GPBE1}.
Using the expression~\eqref{cj formulation 1} of $c_j$ in $\langle H \rangle  $ and adding a constant term
$\beta^{-1} \sum_{j=1}^{N_c}   c_j^\infty\int_\Omega (1-u)        dx$ to adjust the state of zero energy as in \eqref{reference state ensemble} give the polar energy formulation \eqref{ensemble average polar formulation 1}.
From a statistical mechanics perspective, it is indeed important to note that the choice of the ``entropy" $\langle S \rangle$ and ``Helmholtz free energy" $\langle H \rangle$ in the previous derivation is incorrect.
Consequently, \eqref{ensemble average total energy type II} corrects the polar energy formulation within the classical PB theory framework, albeit introducing a discontinuous term 	$\chi_{\{u<1\}}$.
It is crucial to adhere to the correct statistical mechanics principles when formulating the ensemble averages.
Moreover, \eqref{ensemble average solvation energy-steric} further  refines the EASE formulation by accounting for the finite size   effects  and eliminate the discontinuous term $\chi_{\{u<1\}}$.


\section*{Funding Information}

This research was partially supported by the National Science Foundation under grants DMS-2306991 (Shao, Zhao), DMS-1818748 and DMS-2306992 (Chen), DMS-2110914 (Zhao), 
and a CARSCA grant from the University of Alabama (Shao).

\section*{Conflict of interest}

The authors declare that they have no known competing financial interests that could have appeared to influence the work reported in this paper. Although Dr. Shan Zhao is a Journal Editor of the Computational and Mathematical Biophysics, he had no involvement in the peer reviewing process or the final decision of this manuscript.

\section*{Ethics Statement}

This research does not require ethical approval.

\section*{Data availability statement}

Data sharing is not applicable to this article as no data sets were generated or
analysed during the current study.


\bibliographystyle{plain}
\bibliography{refsshao}

\end{document}